\documentclass[11pt,american,english]{article}

\usepackage[latin9]{inputenc}
\usepackage{geometry}
\usepackage[active]{srcltx}
\usepackage{babel}
\usepackage{units}
\usepackage{amsmath}
\usepackage{amsthm}
\usepackage{amssymb}
\usepackage{graphicx}
\usepackage{setspace}
\usepackage[authoryear,round]{natbib}
\PassOptionsToPackage{normalem}{ulem}
\usepackage{ulem}
\usepackage[]
 {hyperref}

\makeatletter
\theoremstyle{plain}
\newtheorem{thm}{\protect\theoremname}
\theoremstyle{plain}
\newtheorem{assumption}[thm]{\protect\assumptionname}
\theoremstyle{plain}
\newtheorem{lem}[thm]{\protect\lemmaname}
\theoremstyle{plain}
\newtheorem{prop}[thm]{\protect\propositionname}
\theoremstyle{plain}
\newtheorem*{lem*}{\protect\lemmaname}
\theoremstyle{plain}
\newtheorem*{prop*}{\protect\propositionname}

\usepackage{babel}
\usepackage{xcolor}
\hypersetup{
    colorlinks,
    linkcolor={red!50!black},
    citecolor={blue!50!black},
    urlcolor={blue!80!black}
}

\ifdefined\showcaptionsetup
 \PassOptionsToPackage{caption=false}{subfig}
\fi
\usepackage{subfig}
\makeatother

\addto\captionsamerican{\renewcommand{\assumptionname}{Assumption}}
\addto\captionsamerican{\renewcommand{\lemmaname}{Lemma}}
\addto\captionsamerican{\renewcommand{\propositionname}{Proposition}}
\addto\captionsamerican{\renewcommand{\theoremname}{Theorem}}
\addto\captionsenglish{\renewcommand{\assumptionname}{Assumption}}
\addto\captionsenglish{\renewcommand{\lemmaname}{Lemma}}
\addto\captionsenglish{\renewcommand{\propositionname}{Proposition}}
\addto\captionsenglish{\renewcommand{\theoremname}{Theorem}}
\providecommand{\assumptionname}{Assumption}
\providecommand{\lemmaname}{Lemma}
\providecommand{\propositionname}{Proposition}
\providecommand{\theoremname}{Theorem}

\begin{document}
\title{\onehalfspacing{}Economic Geography and Structural Change}
\author{Clement E. Bohr, Martí Mestieri and Frédéric Robert-Nicoud\thanks{Clement E. Bohr, UCLA; e-mail: \protect\href{mailto:clement.bohr@anderson.ucla.edu}{clement.bohr@anderson.ucla.edu}.
Martí Mestieri, IAE-CSIC, CREI, BSE, and CEPR; email: \protect\href{mailto:marti.mestieri@iae.csic.es}{marti.mestieri@iae.csic.es}.
Frédéric Robert-Nicoud, GSEM (Université de Genève) and CEPR; e-mail:
\protect\href{mailto:frederic.robert-nicoud@unige.ch}{frederic.robert-nicoud@unige.ch}.
We thank Thibault Fally, Kiminori Matsuyama, Ben Moll, Fidel Perez-Sebastian,
and audiences at Princeton and at the 2024 \textsc{uea} European Meeting
(Copenhagen) for feedback, comments, and suggestions. Mestieri acknowledges
financial support from the Spanish Ministry of Science, Innovation
and Universities (PID2022-139468NB-I00 and Europa Excelencia Grant).
The current paper supersedes our earlier working paper \citep*{BohrMestieriRobertNicoud2023},
which had different focus and title.}}
\date{December 04, 2024}

\maketitle
\medskip{}

\begin{abstract}
\begin{spacing}{0.9}
\noindent As countries develop, the relative importance of agriculture
declines and economic activity becomes spatially concentrated. We
develop a model integrating structural change and regional disparities
to jointly capture these phenomena. A key modeling innovation ensuring
analytical tractability is the introduction of non-homothetic Cobb-Douglas
preferences, which are characterized by constant unitary elasticity
of substitution and non-constant income elasticity. As labor productivity
increases over time, economic well-being rises, leading to a declining
expenditure share on agricultural goods. Labor reallocates away from
agriculture, and industry concentrates spatially, further increasing
aggregate productivity: structural change and regional disparities
are two mutually reinforcing outcomes and propagators of the growth
process. \vspace{1cm}


\medskip{}

\end{spacing}
\end{abstract}

\section{\protect\label{sec:Introduction}Introduction}

The geographic and sectoral distributions of economic activity jointly
evolve along the development path. Figure \ref{fig:intro} shows how,
as income per capita grows, the share of value added in agriculture
declines (Figure \ref{fig1a}) and the share of population in urban
areas increases (Figure \ref{fig1b}) for a panel of 38 countries
spanning very different income levels. Combining these two patterns,
Figure \ref{fig1c} shows a strong, negative correlation between agricultural
value added shares and spatial concentration of economic activity
as proxied by urbanization rates. What is the relationship between
structural change and spatial concentration? Does one drive the other
or are they jointly determined? Are there any feedback effects between
them?

The purpose of this paper is to provide a parsimonious framework to
analyze these questions, bringing together elements of economic geography
and structural change. In particular, we combine an economic geography
framework with a demand-driven theory of structural change in which
we introduce a novel non-homothetic demand system -- ``Heterothetic
Cobb-Douglas'' (henceforth \textsc{hcd}). These preferences feature
a variable income elasticity of demand while maintaining a constant,
unitary elasticity of substitution. This combination of elements (particularly,
the unitary elasticity of substitution) results in a tractable unified
framework that enables us to study a rich set of two-way interactions
between structural change and the evolution of regional disparities.\footnote{The observation that urbanization and structural change are broadly
correlated is well-known and different facets of this relationship
have been explored. See, among others, \citet*{CaselliColeman2001,EckertPeters2022,FajgelbaumRedding2022,MichaelsRauchRedding2012,Nagy2023}
and the discussion in the related literature. However, as we explain
below, the extant literature has not addressed our research question.} Our central result is to show how rising incomes yield spatial concentration
through structural change, and how rising spatial concentration increases
incomes, fueling structural change. Importantly, our framework allows
us to show how both structural change and spatial concentration arise
as an outcome of economic growth, without needing a fall in transportation
and trade costs. The theory can also shed light on the evolution of
earnings inequality.

\begin{figure}
\caption{The Joint Evolution of Agricultural and Urban Economic Activity}
\label{fig:intro}

\subfloat[Agricultural Shares and Income pc]{\label{fig1a}

\includegraphics[scale=0.36]{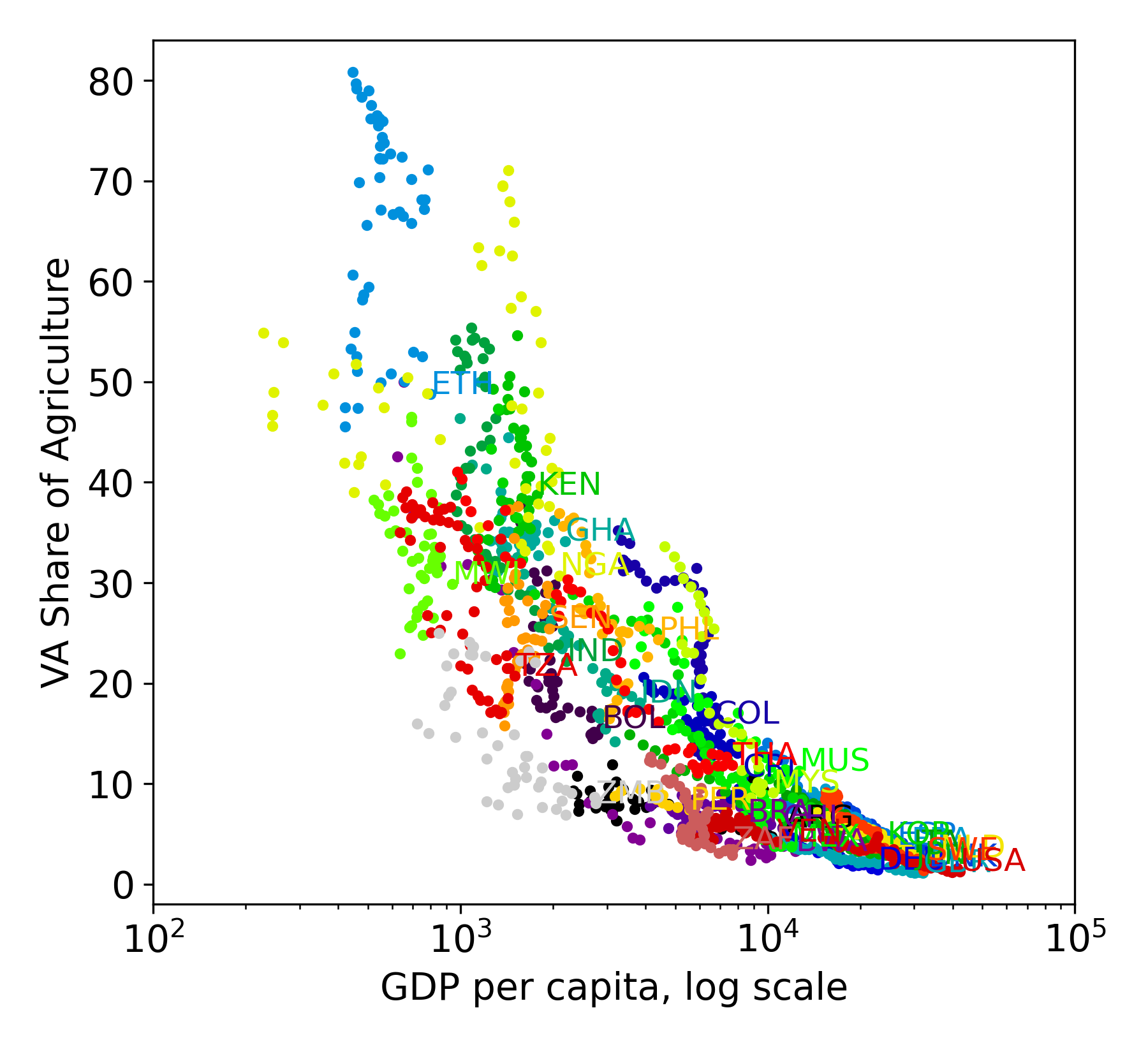}}\subfloat[Urbanization and Income pc]{\label{fig1b}

\includegraphics[scale=0.36]{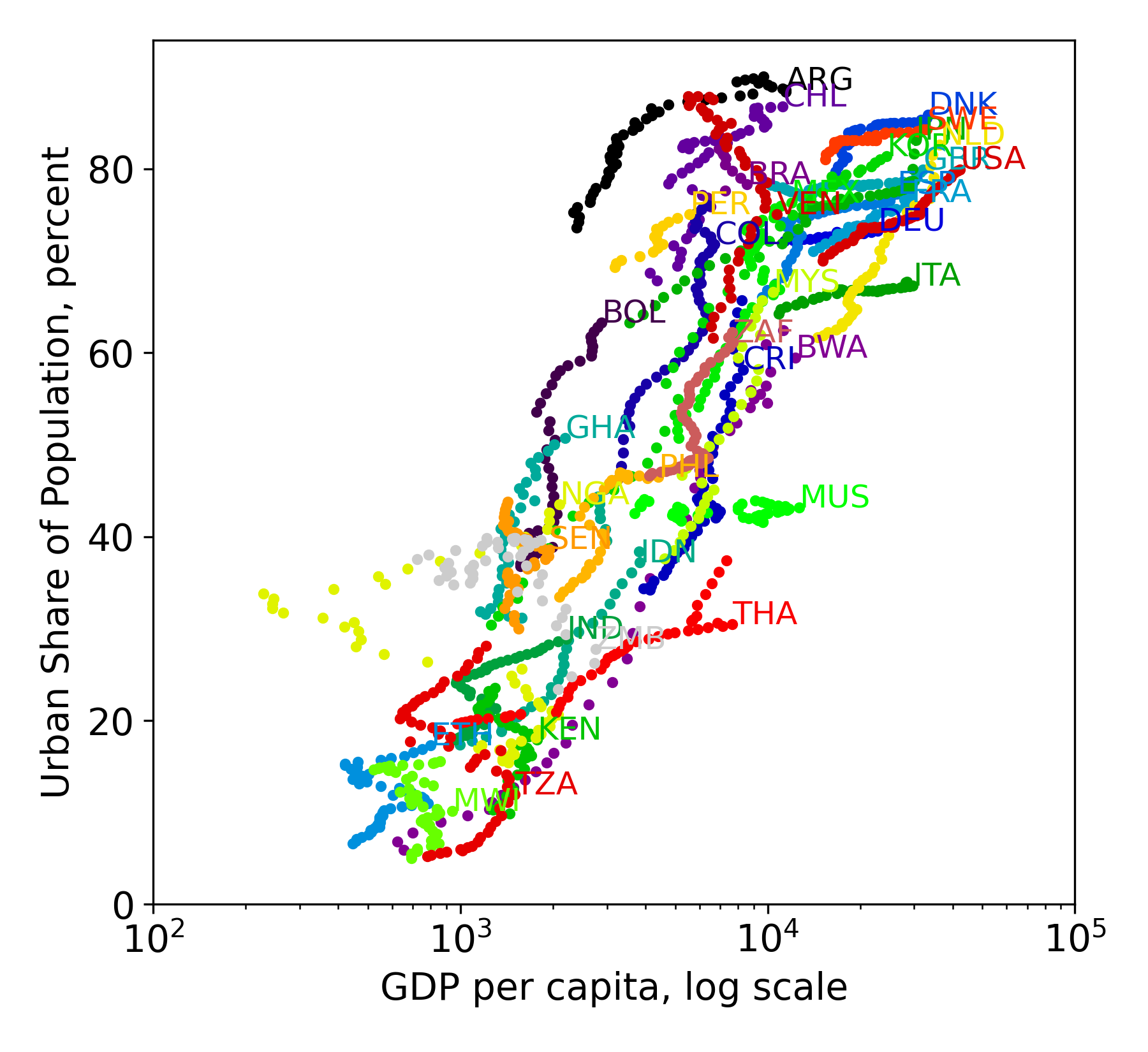}}\subfloat[Agri. Shares and Urbanization]{\label{fig1c}

\includegraphics[scale=0.35]{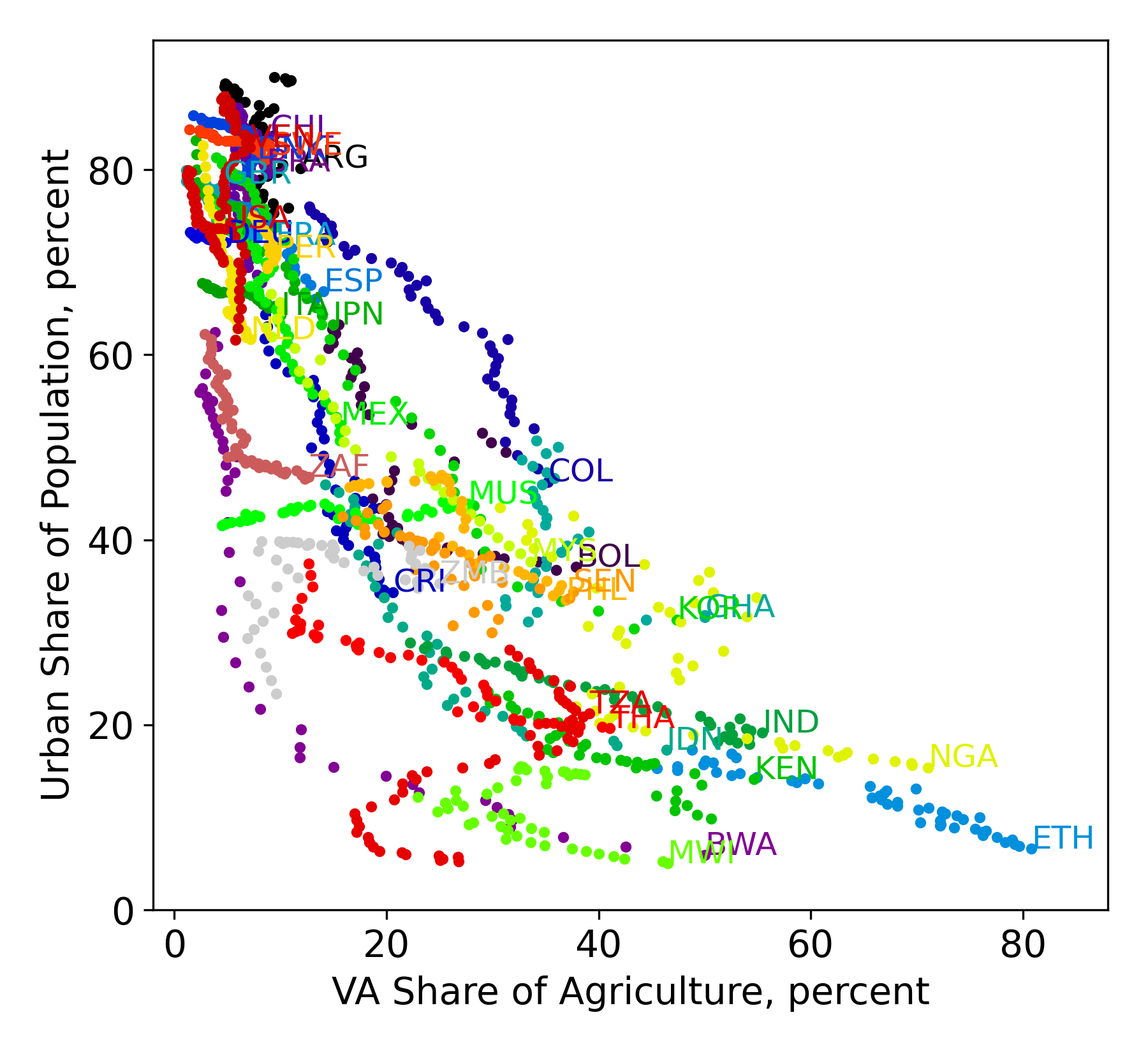}}

{\footnotesize\emph{Notes}: Data for agricultural value-added shares
comes from the Groningen 10-sector database as computed in \citet*{CominLashkariMestieri2021}.
Urbanization rates comes from the \href{https://data.unhabitat.org/}{United Nations Urban Indicators Database}.
Income per capita corresponds to 2017 US dollars (PPP-adjusted) from
the PWT. See Appendix \ref{sec:Data-Construction} for more details.}{\footnotesize\par}
\end{figure}

For simplicity, we consider only two factors of production, two sectors,
and two regions endowed with identical fundamentals, as in the seminal
``new'' economic geography (henceforth \textsc{neg}) paper by \citet{Krugman1991}.
The factors are labor (owned by ``workers'') and human capital (owned
by ``entrepreneurs''). We label the sectors agriculture and manufacturing,
where manufacturing lumps together the manufacturing and service industries.
In line with empirical evidence, the income elasticity of demand for
this composite good is larger than one, hence its share in national
income increases as real income per capita grows over time. Conversely,
agriculture is a necessity, and its share in employment and in national
income falls over time: the economy undergoes a structural transformation
(see, e.g., \citealp{CaselliColeman2001}, and \citealp*{Kongsamut2001}).\footnote{\citet*{Kongsamut2001} define structural change (or structural transformation)
as ``the massive reallocation of labor from agriculture into manufacturing
and services that accompanies the growth process {[}p. 869{]},''
a process to which \citet{Matsuyama2016} refers as the ``Generalized
Engel's Law.'' See \citet*{Herrendorf2014} for a comprehensive synthesis
of the field.}

In our model, agriculture is tied to the land, and the workers it
employs are a source of geographically immobile demand for the composite
good; this so-called dispersion force tends to disperse economic activity
across space. Conversely, entrepreneurs are mobile and the manufacturing
firms they manage tend to move to the region with the larger demand
to save on trade and transportation costs; in turn, the region with
the larger demand tends to be the one hosting the most entrepreneurs
since they also consume manufacturing goods. Thus, the larger their
expenditure share on manufacturing, the stronger this positive feedback
loop, and the more likely these self-enforcing agglomeration economies
dominate dispersion forces leading to the spatial concentration of
manufacturing in equilibrium. This expenditure share in manufacturing
is a key \emph{parameter} in Krugman's core-periphery model. In our
model, by contrast, it is an \emph{endogenous outcome} that depends
on the income level of the economy. As we explain below, unlike the
fall of transportation costs, whose impact on regional disparities
depends on subtle modeling assumptions, the rise of labor productivity
and of the equilibrium expenditure share on manufacturing has an unambiguous
impact on the emergence of regional disparities.\footnote{In \citet{Krugman1991}, the source of immobile expenditure comes
from agriculture, a competitive sector producing under \emph{constant}
returns to (spatially immobile) labor and whose output is \emph{freely
traded}; Krugman shows that regional disparities arise in equilibrium
only if transportation costs are \emph{low} enough. In \citet{helpman1998size},
the source of immobile expenditure comes from housing, a competitive
sector supplying a \emph{non-traded} service; Helpman shows that regional
disparities arise in equilibrium only if transportation costs are
\emph{high} enough. Finally, in \citet{KrugmanVenables1995}, the
source of immobile expenditure comes from agriculture, a competitive
sector producing under \emph{decreasing} returns to labor and whose
output is \emph{freely traded}; they show that regional disparities
arise for \emph{intermediate} levels of trade costs.} In particular, regional disparities are an equilibrium outcome regardless
of the level of transportation costs if labor productivity and the
equilibrium expenditure share on manufacturing are high enough. The
endogenous emergence of such a ``black hole'' cannot arise in original
\textsc{neg} models featuring homothetic preferences and constant
expenditure shares.

Our framework highlights the two-way relationship between regional
disparities and structural change: growth-led structural change generates
regional disparities, and the emergence of regional disparities induces
further structural change. Consider the following thought experiment.
Assume that labor productivity is initially low: the economy is predominantly
agricultural and manufacturing firms are evenly dispersed across regions
at the unique stable equilibrium. As labor productivity monotonically
increases over time, per capita incomes and economic well-being rises.
In turn, due to the non-homothetic demand system, the share of spending
on manufactures increases and that on agriculture falls. The first
implication of this process is that labor reallocates from agriculture
to manufacturing: \emph{structural change} unfolds.

As the share of spending on manufactures increases, a second implication
is that agglomeration forces strengthen and dispersion forces weaken
since entrepreneurs benefit from locating in the same region. Eventually,
the spatial concentration of manufacturing becomes the only stable
equilibrium outcome: \emph{regional disparities} emerge. This result
illustrates how structural change fuels regional disparities, and
it is thus unlikely to happen uniformly across space. Industrialization
may only take place in a few ``cores'' \citep{PugaVenables1996}.
Furthermore, the spatial concentration of manufacturing exerts a positive
feedback effect on income growth and structural change. As industry
concentrates geographically, the price of the consumption basket of
manufacturing falls for the entrepreneurs -- since by consuming in
the same location of production, they minimize the trade and transportation
costs. As a result, the utility levels of entrepreneurs increase and,
since manufacturing is a luxury, their expenditure share on manufacturing
increases. This mechanism illustrates how spatial concentration of
manufacturing fuels income growth and structural change.

Finally, structural change brought about by the combination of steady
growth of labor productivity and of manufacturing being a luxury good
also has an effect on income inequality. Though economic growth benefits
both types of factor owners, productivity growth disproportionately
benefits the owners of the factor used intensively in the sector growing
the most -- entrepreneurs. This logic is well-known since at least
\citet{Jones1965}, and here it extends to the case of a factor-specific
model featuring monopolistic competition.

\section{Relation to the Literature}

Our paper combines elements from the fields of economic geography
and from structural change.

The first major ingredient of our model borrows from the arguably
no-longer aptly named ``new'' economic geography. It features two
regions with identical fundamentals, some geographically mobile factors,
monopolistic competition, scale economies, and trade and transportation
costs, as in \citet{Krugman1991} and in the vast literature that
followed \citep*{BaldwinForslidMartinOttavianoRobertnicoud,FujitaKV1999}.
Under some conditions, the combination of these ingredients gives
rise to self-enforcing agglomeration economies and to the emergence
of regional disparities. The basic structure of Krugman's model is
simple but its analysis is complex, and it took fifteen years and
combined efforts of several authors to complete a comprehensive analysis
of standard \textsc{neg} models.\footnote{\citet{Krugman1991} provides necessary conditions for the core-periphery
pattern to be a stable equilibrium outcome but uses numerical simulations
otherwise. Even then, he writes that this necessary condition ``appears
to be a fairly unpromising subject for analytical results. However,
it yields to careful analysis {[}p. 495{]}.'' The same can be said
about the whole model. First, \citet{Puga1999} developed the analytical
methods to provide sufficient conditions for a core-periphery equilibrium
to be stable. Second, \citet{RobertNicoud2005} showed how the main
\textsc{neg} models are isomorphic, and used that result to establish
the overall stability properties of these models and to characterize
the universe of equilibriums. Finally, \citet{Mossay2006} established
the existence and uniqueness of the so-called short run equilibrium
of Krugman's model. \citet*{CharlotGRNT2006} characterize the normative
properties of Krugman's model.} Even though we add one source of complexity to this framework --
non-homothetic preferences -- our model does so in a relatively parsimonious
way. It encompasses the variant of Krugman's model introduced by \citet{Forslid2003}
as a special case (when preferences are homothetic Cobb-Douglas).
Our paper also complements the works of \citet*{BaldwinMartinOttaviano2001},
who study how innovation and growth interact in an \textsc{neg} framework,
and of \citet*{OttavianoTT2002} and \citet{Pflueger2004}, who introduce
\textsc{neg} models with quasi-linear (and hence non-homothetic) preferences.
Both models are analytically more amenable than Krugman's original.\footnote{\citet*{OttavianoTT2002} introduce a model of monopolistic competition
that features true pro-competitive effects, and they show how most
results derived by \citet{Krugman1991}, who uses the famous framework
developed by \citet{DixitStiglitz1977}, are robust to this alternative
functional form.} However, manufacturing is a necessity and agricultural goods are
a luxury in these models, which is counterfactual.

Krugman's original focus, as well as that of much of the \textsc{neg}
literature's, was on the role of trade and transportation costs in
fostering regional disparities. This focus resonated to Europeans
in particular, as it shed light on the European Single Market Program
from an original angle at the time. The economic mechanisms unveiled
by the \textsc{neg} suggest that globalization and European economic
integration may contribute to deepening of regional disparities rather
than fostering spatial convergence \citep{Puga2002}. In this paper,
we shift the narrative by demonstrating that falling transportation
costs are neither necessary nor even likely to be the main driver
of regional disparities. Rather, our analysis emphasizes that structural
change associated with economic growth induces regional disparities
on its own.

This approach addresses two critiques of Krugman\textquoteright s
original model. First, \citet{helpman1998size} includes consumption
of housing (a non-traded service) instead of agriculture (a freely
traded good in Krugman's model) in his alternative \textsc{neg} model.
Helpman shows that competition among mobile workers for this non-tradable
service contributes to regional convergence, and that falling transportation
costs actually amplify this dispersion force: trade integration does
not cause regional disparities; quite the opposite, it contributes
to regional convergence.\footnote{\citet{KrugmanVenables1995} develop a model with decreasing returns
to labor in manufacturing, which leads to regional divergence for
intermediate levels of transportation costs and to convergence for
high and low trade frictions. We extend our model in this direction
in Appendix \ref{Appx:Helpman}.} Our theory is robust to Helpman's critique, as increasing agglomeration
forces are driven by a rising expenditure share in the good that exhibits
increasing returns to scale -- manufacturing -- and this force persists
even when land is included, as long as housing is a necessity (which
is empirically the case).\footnote{Housing being a necessity, its expenditure share is declining in income.
A rise in labor productivity therefore reduces the dispersion force
as households choose to spend a smaller share of their income on it.}

Second, Krugman himself nuanced the findings of his paper, by emphasizing
that his model features a region of the parameter space (when the
expenditure share on the sector featuring economies of scale is sufficiently
large) that implies the emergence of regional disparities in equilibrium
regardless of the level of transportation costs. He addressed this
caveat by introducing the ``No Black-Hole Condition'' ruling out
this outcome by assumption. From the perspective of our paper with
endogenous expenditure shares, the black hole becomes an equilibrium
outcome, and therefore a feature of the model rather than a bug. As
labor productivity rises and the expenditure share on manufacturing
consequently increases, the economy transitions from having no regional
disparities, to having the potential of regional disparities, to inevitably
featuring regional disparities in equilibrium, i.e. the black hole.
An important implication of this results is that prohibitive transportation
costs in a sector are not inconsistent with its spatial clustering.
Our model can thus accommodate the agglomeration of non-traded services
as an equilibrium outcome.

The second major ingredient of our model is demand-driven structural
change: expenditure and resources reallocate away from agriculture
as economic well-being rises. \citet{dennisiscan09} document a preponderant
role of income effects in the United States' structural change out
of agriculture from 1800 to 1950.\footnote{For the postwar United States, when the agricultural share is already
small, \citet*{HerrendorfRogersonValentinyi2013}, \citet{Boppart2014}
and \citet*{CominLashkariMestieri2021} among others find that both
income and price effects are important for explaining structural transformation
using different partitions of the economy and demand specifications.} We deliberately switch-off the relative price-effect channel by setting
the elasticity of substitution between manufacturing and agriculture
to one. Thus, for any given fixed level of utility, any form of technical
progress on the supply side leaves expenditure shares and employment
shares unchanged. It would be conceptually straightforward to extend
the analysis of our paper to allow this channel to operate by using
the non-homothetic constant elasticity of substitution (\textsc{nh-ces})
to model structural change \citep*{CominLashkariMestieri2021} in
a spatial context as in \citet{FinlayWilliams2022} and \citet{Takeda2022}.
Such an extension is ancillary to the main point we make here: economic
growth is a source of both structural change and regional disparities
when preferences are non-homothetic.\footnote{A series of papers study aspects of structural change and economic
geography that work through the relative price-effect channel; these
authors use homothetic preferences. \citet*{desmet2014spatial} develop
a spatial model of innovation and technology diffusion. \citet*{FajgelbaumRedding2022}
analyze the role of access to internal and external markets in the
structural transformation of late 19th-century Argentina. \citet*{MichaelsRauchRedding2012}
and \citet{Nagy2023} analyze the joint process of urbanization and
structural transformation in the United States starting in the later
19th century. Using data from Brazil, \citet*{Bustos2016} show that
the effect of agricultural productivity growth on structural change
depends on the factor-bias of the technical change.}

Our paper contributes to the growing literature that studies structural
change and economic geography jointly. A few of these papers also
use non-homothetic preferences. The closest paper to ours is \citet{Murata2008},
who uses Stone-Geary preferences \citep{geary1950note,stone1954linear}
in an otherwise standard two-region \textsc{neg} model. His analysis
focuses on the role of trade and transportation costs, as do most
\textsc{neg} contributors. \citet{CaselliColeman2001} show that structural
transformation is a source of regional income \emph{convergence} in
neoclassical models in the presence of decreasing marginal returns
of mobile factors of production and in the absence of trade costs
among regions. They also use Stone-Geary preferences, which feature
non-unitary income and price elasticities of demand. By contrast,
\citet*{ChatterjeeGiannoneKleinebergKuno2023} document that regional
convergence has stalled since the 1980s for a panel of 34 countries
and relate it to the rise of the service sector. \citet{EckertPeters2022}
and \citet{BudiOrsPijoanMas2022} build quantitative spatial general
equilibrium models to study the spatial dimension of structural change
in the United States at the turn of the twentieth century \citep{EckertPeters2022}
and in Spain from 1940 to 2000 \citep{BudiOrsPijoanMas2022}. \citet{CaselliColeman2001},
\citet{BudiOrsPijoanMas2022} and \citet*{ChatterjeeGiannoneKleinebergKuno2023}
use Stone-Geary preferences, which, like all explicitly additively
separable preferences, display fundamental flaws \citep{Matsuyama2016}.
The \textsc{hcd} preferences we use to model structural change do
not suffer from these flaws. \citet{EckertPeters2022} follow \citet{Boppart2014}
in using price-independent generalized linear (\textsc{pigl}) preferences.
These preferences have more desirable properties than Stone-Geary
preferences while maintaining aggregation, but they only accommodate
two sectors. By contrast, our \textsc{hcd} preferences are a limiting
case of \textsc{nh-ces} preferences; they accommodate an arbitrary
number of goods and sectors and have a constant unitary price elasticity.
They also yield a ``generalized separable'' demand system \citep{Fally2022,Gorman_1995,Pollak1972}.
We anticipate that \textsc{hcd} preferences will be of independent
interest and applied to a wide range of contexts.\footnote{For a discussion on potential applications, see \citet*{BohrMestieriRobertNicoud2023}
and \citet*{Bohr2023}.}

\section{\protect\label{sec:hcd}Heterothetic Cobb-Douglas Preferences}

We begin by characterizing \textsc{hcd} preferences in their general
form after which we introduce them to the NEG model of spatial and
structural change. All proofs are contained in the appendix.

Consider the utility function $U:\times_{i\in I}\left[\gamma_{i},\infty\right)\mapsto\left(1,\infty\right)$
that maps the consumption vector $\left\{ C_{i}\right\} _{i\in I}$
into a level of utility $u$, where $u$ is implicitly defined by:\footnote{Note that it is conceptually straightforward, but ancillary to the
main point of this paper, to extend the analysis to \textsc{nh-ces}
preferences over $C_{i}$, with an elasticity of substitution $\eta\equiv\frac{1}{1-\rho}$
that is different from one and smaller than $\sigma$. Specifically,
assume that $u$ is defined implicitly as follows:
\[
\frac{u^{\rho}-1}{\rho}=\sum_{i}\omega_{i}\left(u\right)\frac{\left(C_{i}\right)^{\rho}-\gamma^{\rho}}{\rho},
\]
where the functions $\omega_{i}$ satisfy assumptions \ref{assu:omega}.
Then, we can write the Heterothetic Cobb-Douglas preferences in equation
\eqref{eq:HCD} as the limiting case of $\rho\rightarrow0$ for these
\textsc{nh-ces} preferences.}

\begin{equation}
\ln u=\sum_{i}\omega_{i}\left(u\right)\ln\left(\frac{C_{i}}{\gamma_{i}}\right)-\sum_{i}\omega_{i}\left(u\right)\ln\left(\frac{\omega_{i}\left(u\right)}{\omega(u)}\right),\label{eq:HCD-1}
\end{equation}
and where $\omega(u)\equiv\sum_{i}\omega_{i}(u)$. The parameter $\gamma_{i}$
defines some minimal amount of good $i$ over which $U$ is defined.
The distinction from homothetic Cobb-Douglas is that the weights $\omega_{i}$
are a no longer a parameter, but a function of the utility level $u$.
We impose the following regularity conditions on these weight functions,
which ensure interior solutions for the utility maximization program,
and in turn that the expenditure and indirect utility functions are
well-behaved:\footnote{See \citealp{Fally2022}, for a crisp and exhaustive analysis of the
conditions that ensure that the expenditure and indirect utility functions
are well behaved}
\begin{assumption}
\label{assu:omega}\textbf{Self-Consistency}. There exist real numbers
$\{\text{\underbar{s}}_{i}\}\in(0,1)$, such that $\forall i\in I,\forall u\geq1$:
(i) $\text{\underbar{s}}_{i}\leq\frac{\omega_{i}(u)}{\omega(u)}$,
(ii) $\partial\omega_{i}/\partial u\leq0$ (with strict inequality
for at least one good), and (iii) the vector of prices and income
$\left(\left\{ P_{i}\right\} ,y\right)$ obeys 
\begin{equation}
P_{i}\gamma_{i}<\uline{s}_{i}y.\label{eq:sufficient=000020condition-1}
\end{equation}
\end{assumption}

Parts (i) and (iii) of this assumption together ensure that the share
of income that the preferences imply should be allocated to any given
good exceeds the minimum required level of expenditures for that good
$P_{i}\gamma_{i}$. Without the assumption, the preferences may be
internally inconsistent, requiring a level of consumption on a good
that is higher than the expenditure share the preferences imply should
be allocated toward it. The assumption also implies a limit on the
realized variance of prices, which are an endogenous outcome. Generally,
one would have to make assumptions on fundamentals that ensure that
this inequality is verified in equilibrium, and we will highlight
the needed parameter restrictions below (see Lemma \ref{lem:Sufficient=000020condition}).
In the limiting case of $\gamma_{i}\rightarrow0$, however, this inequality
is verified as long as prices remain finite. This case imposes the
most minimalist residual restrictions, and we therefore make the following
assumption:
\begin{assumption}
\label{assu:subsistence}\textbf{Regularity conditions}. $\forall i\in I:\gamma_{i}\in\mathbb{R}_{++}$.
We choose units of measurement for these goods so that $\gamma_{i}=\gamma$
for all $i\in I$, and let $\gamma\rightarrow0$, in which case $U$
is defined almost everywhere. For simplicity, we will also assume
that $\uline{s}_{i}=\uline{s}$ for all $i\in I$.
\end{assumption}

\paragraph{Existence and global convexity of the function $U$.}

Note that the requirement $C_{i}\geq\gamma$ and assumptions \ref{assu:omega}
and \ref{assu:subsistence} together imply that $\ln\left(\frac{\omega C_{i}}{\omega_{i}\gamma}\right)>0$
holds for any good $i$, and that the right-hand side of equation
\eqref{eq:HCD-1} is decreasing in $u$. We obtain the following result.
\begin{lem}
\label{lem:Regularity=000020conditions=000020--=000020consumption}\textbf{Regularity
conditions}. Assume preferences obey equation \eqref{eq:HCD}, and
assumptions \ref{assu:omega} and \ref{assu:subsistence} hold; then
in equilibrium (i) $C_{i}>\gamma$ and (ii) $\nicefrac{\partial u}{\partial C_{i}}>0$
for all $i\in I$.
\end{lem}

\paragraph{Indirect utility and expenditure functions.}

The indirect utility function $V:\times_{i}\left(0,\infty\right)\mapsto\left(1,\infty\right)$,
maps the vector of prices and income $\left(\left\{ P_{i}\right\} ,y\right)$
into a level of utility. Under assumptions \ref{assu:omega} and \ref{assu:subsistence},
which we maintain throughout the paper as they ensure an interior
solution, the optimal consumption allocations under \textsc{hcd} preferences
imply relative expenditure shares that are proportional to the the
relative preference weights, as is the case under homothetic Cobb-Douglas.
We can therefore define and write expenditure shares as 
\begin{equation}
m_{i}\left(u\right)\equiv\frac{P_{i}C_{i}}{y}=\frac{\omega_{i}\left(u\right)}{\omega\left(u\right)}.\label{eq:expenditure=000020share}
\end{equation}

To obtain the indirect utility function, plug $C_{i}=\frac{\omega_{i}y}{\omega P_{i}}$
and $v=u$ into \eqref{eq:HCD-1} to get:
\begin{equation}
\ln v=\sum_{i}\omega_{i}\left(v\right)\ln\left(\frac{y}{\gamma P_{i}}\right).\label{eq:Indirect=000020Utility=000020--=000020intermediate=000020step}
\end{equation}
The derivative of the left hand side in the expression above with
respect to $v$ is positive by inspection, while the derivative of
its right-hand side is negative by assumptions \ref{assu:omega} and
\ref{assu:subsistence}. It then follows from $\omega_{i}>0$ that
$v$ is increasing in $y$ and decreasing in $P_{i}$. We can then
establish the following results:
\begin{lem}
\textbf{\label{lem:Indirect=000020Utility}Indirect utility function}.
Assume assumptions \ref{assu:omega} and \ref{assu:subsistence} hold.
Then the indirect utility $V\left(\{P_{i}\},y\right)$ associated
with the preferences in equation \eqref{eq:HCD-1} can be written
implicitly as the fixed point for $v$ of expression \eqref{eq:Indirect=000020Utility=000020--=000020intermediate=000020step}.
This fixed point exists and is unique. Furthermore, $V$ is increasing
in $y$, decreasing in $P_{i}$, and homogeneous of degree zero in
its arguments. Hence $V$ is a proper indirect utility function.
\end{lem}

\begin{lem}
\textbf{\label{lem:Expenditure=000020Function}Expenditure function}.
Assume assumptions \ref{assu:omega} and \ref{assu:subsistence} hold.
Then the expenditure function $E\left(u,\{P_{i}\}\right)$ associated
with the preferences in equation \eqref{eq:HCD-1} can be written
as:
\begin{align}
\ln E\left(u,\{P_{i}\}\right) & =\ln\Omega\left(u\right)+\sum_{i}m_{i}\left(u\right)\ln\left(P_{i}\right),\label{eq:Expenditure=000020Function=000020E}
\end{align}
where 
\begin{equation}
\ln\Omega\left(u\right)\equiv\ln\gamma+\frac{\ln u}{\omega\left(u\right)}.\label{eq:Omega}
\end{equation}
The function $E$ is increasing in its arguments and it is homogeneous
of degree one in prices. Hence $E$ is a proper expenditure function.
\end{lem}

Equating expenditure with nominal income $y$, it follows from equation
\eqref{eq:Expenditure=000020Function=000020E} that we can naturally
interpret $\Omega$ as the real income: $\Omega=y/\mathbb{P}$, where
$\ln\mathbb{P}\equiv\sum_{i}m_{i}\left(u\right)\ln\left(P_{i}\right)$
is a (utility-dependent) price index. Notably, this price index is
computable with readily observable data: expenditure shares and prices.
It follows from $u>1$ and $\omega^{\prime}<0$ that $\Omega$ is
monotonically increasing in $u$. Since utility is an ordinal concept,
we can interchangeably use $u$ or $\Omega$ as our metric for utility.
We will make use of this property in section \ref{sec:Equilibrium}
below.

\section{\protect\label{sec:Model}A Two-Region Model of Spatial and Structural
Change}

After having introduced \textsc{hcd} preferences, we now present our
application to a model of spatial and structural changes. Our full
model is a succession of static equilibriums in which growth is driven
by exogenous increases in labor productivity. We begin by presenting
the static equilibrium in this section and next. To maximize analytical
tractability, the static model follows the structure of \citet{Forslid2003},
a more analytically tractable variant of Krugman's seminal 1991 paper,
as closely as possible.

\paragraph{Model Overview}

The economy consists of two regions, $r,s\in\left\{ 1,2\right\} $,
and two sectors, agriculture and manufacturing $i\in\{A,M\}$. Agricultural
production is a constant returns to scale sector that is tied to the
land. It produces a homogeneous good competitively and its output
is traded at no cost between the two regions. Manufacturing consists
of differentiated goods produced under increasing returns that are
internal to each firm. Each firm produces a different variety. Market
conduct is monopolistically competitive as in \citet{DixitStiglitz1977}.
Manufacturing firms may locate in either region, and their output
is traded at a cost, which takes the standard iceberg form, as in
\citet{Samuelson1954}. There are two factors of production, labor
(workers), and human capital (entrepreneurs). Workers are geographically
immobile but they may work in either sector, as in \citet{Forslid2003}.
We normalize the mass of workers to one, so there is a mass of $\frac{1}{2}$
workers in each region. We denote their wage by $w^{r}$. Entrepreneurs
manage manufacturing firms and may move across regions in search of
the highest real returns. We also normalize the mass of entrepreneurs
to one, and denote by $\lambda^{r}$ the share of entrepreneurs who
live and produce in region $r$ and denote their nominal earnings
by $\pi^{r}$. Finally, we denote the prices of the agriculture good
and of the composite manufacturing good by $P_{A}$ and $P_{M}$,
respectively.\footnote{Note that two of our assumptions differ from those made by \citet{Krugman1991}.
These two assumptions -- that workers consume manufactures only and
that workers are employed in both sectors for a wage that is constant
in equilibrium -- are made for analytical convenience, and the equilibrium
implications of these differences are innocuous. \citet*{BaldwinForslidMartinOttavianoRobertnicoud}
show that a wide array of new economic geography models share the
same equilibrium properties. \citet{RobertNicoud2005} formally shows
that these models are actually \emph{isomorphic} in an adequately
chosen state space.}

\subsection{Production and Trade Technologies}

\paragraph*{Agricultural Technology.}

The production technology for agriculture is linear in labor $A^{r}=\alpha L_{A}^{r},$
where $A^{r}$ and $L_{A}^{r}$ denote output and labor input in region
$r$, and $\alpha>1$ is the marginal product of labor in agriculture.

\paragraph*{Manufacturing Technology.}

Each manufacturing firm produces a differentiated variety and is run
by a single manager. Labor is the sole variable input. Let $\mu>0$
denote the marginal product of labor in manufacturing and $\pi^{r}$
the nominal earnings of the entrepreneur running the firm. The cost
function of an arbitrary manufacturing firm established in region
$r$ and producing $x$ units of output is 
\begin{equation}
c^{r}\left(x\right)=\pi^{r}+\frac{w^{r}}{\mu}x.\label{eq:Cost=000020fcn}
\end{equation}

\paragraph{Trade Technology.}

Agricultural output is freely traded and is produced in both regions.
By contrast, trade in manufactures is costly. Firms may ship their
output at no cost within their own region, but need to ship $\tau>1$
units of output for one unit to arrive in the other region (the rest
melts in transit). We henceforth refer to $\tau$ as ``trade costs''
for short, which encapsulates any physical or legal cost borne from
conducting inter-regional trade.

\subsection{Endowments and Preferences}

\paragraph{Workers.}

Each worker is endowed with one unit of labor that is inelastically
supplied. Workers consume manufacturing goods only. The utility they
derive from consuming manufacturing varieties is a constant-elasticity-of-substitution
(CES) aggregate $C_{M}$, 
\begin{equation}
C_{M}=\left[\int_{0}^{N}\left(c\left(n\right)\right)^{\frac{\sigma-1}{\sigma}}dn\right]^{\frac{\sigma}{\sigma-1}},\label{eq:C_M}
\end{equation}
where $\sigma>1$ is the elasticity of substitution, and $N$ is the
endogenous mass of active firms, which is pinned down in equilibrium.

\paragraph{Entrepreneurs.}

Entrepreneurs are endowed with one unit of human capital, which they
use to run firms. Entrepreneurs hold two-tiered preferences. The upper
tier is defined over the agricultural and composite manufacturing
good, and the lower tier is defined over manufacturing varieties with
the same \textsc{ces} aggregator $C_{M}$ in equation \eqref{eq:C_M}
as for workers. For the upper tier, entrepreneurs have \textsc{hcd}
preferences over agriculture and manufacturing as laid out in the
previous section. The utility function $U:\times_{i\in\left\{ A,M\right\} }\left[\gamma,\infty\right)\mapsto\left(1,\infty\right)$
maps the consumption vector $\left\{ C_{A},C_{M}\right\} $ into a
level of utility $u$ implicitly defined by

\begin{equation}
\ln u=\omega_{A}\left(u\right)\ln\left(\frac{C_{A}}{\gamma}\right)+\omega_{M}\left(u\right)\ln\left(\frac{C_{M}}{\gamma}\right)-\sum_{i\in A,M}\omega_{i}(u)\ln\left(\frac{\omega_{i}(u)}{\omega(u)}\right).\label{eq:HCD}
\end{equation}
As discussed in Section \ref{sec:hcd}, the associated expenditure
function is
\begin{align}
\ln E\left(P_{A},P_{M},u\right) & =\ln(\Omega\left(u\right))+m\left(u\right)\ln\left(P_{M}\right)+(1-m(u))\ln(P_{A}).\label{eq:Expenditure=000020Function=000020E-2}
\end{align}

In order to have the manufacturing composite good be a luxury and
the agricultural good be a necessity, we impose the additional assumption:
\begin{assumption}
\label{assu:Engel}\textbf{Manufacturing is a luxury}. $\frac{\omega_{M}}{\omega_{A}}$
is increasing everywhere in $u$.
\end{assumption}

It follows from assumptions \ref{assu:omega}(i) and \ref{assu:Engel}
that the expenditure share on manufactures belongs to the interior
of the unit interval and is increasing in $u$:
\begin{equation}
\forall u\in\left(1,\infty\right):\qquad0<m\left(u\right)<1,\quad0<m^{\prime}\left(u\right).\label{eq:expenditure=000020share=000020range}
\end{equation}

\subsection{Markets and Equilibrium conditional on Location Choices}

\paragraph{Agriculture.}

The market for the agricultural good is competitive. Hence, in any
equilibrium the law of one price holds. The agricultural good is then
a natural choice for numeraire, and we set $P_{A}=1$. Marginal cost
pricing and the absence of trade and transportation costs together
imply that
\begin{equation}
w^{r}=\alpha.\label{eq:w^r=00003Dalpha}
\end{equation}

Henceforth, we refer to parameter $\alpha$ as ``labor productivity.''
An increase in $\alpha$ rises labor productivity in the agricultural
sector and, ultimately, the equilibrium wage for labor in the economy.
Labor productivity will be a central parameter of interest in the
analysis of spatial structural change.

\paragraph{Manufacturing.}

Manufacturing operates under monopolistic competition. Markets are
segmented and firms set prices in both regions $r\in\left\{ 1,2\right\} $
independently. Since all firms are managed by exactly one manager,
all managers are fully employed in equilibrium, hence $N=1$. The
combination of \textsc{ces} preferences among manufactures on the
demand side and monopolistic competition with segmented markets on
the supply side gives rise to the equilibrium f.o.b. prices $p^{r}=\frac{\sigma}{\left(\sigma-1\right)\mu}w^{r}$.
Without loss of generality, we set $\mu=\frac{\alpha\sigma}{\sigma-1}$
so that the equilibrium f.o.b. price is equal to\footnote{Let $p\left(i,n\right)$ denote the consumer price that individual
$i$ faces when purchasing the manufacturing variety produced by firm
$n$, and let $P_{M}\left(i\right)$ be the price-index of the manufacturing
bundle that they consume. The Marshallian demand of a consumer maximizing
sub-utility $C_{M}$ given that they allocate an arbitrary expenditure
amount $e_{M}\left(i\right)$ on manufactures is equal to
\[
d\left(i,n\right)=\left(\frac{p\left(i,n\right)}{P_{M}\left(i\right)}\right)^{-\sigma}\frac{e_{M}\left(i\right)}{P_{M}\left(i\right)},\qquad P_{M}\left(i\right)=\left[\int_{0}^{N}\left(p\left(i,n\right)\right)^{1-\sigma}dn\right]^{\frac{1}{1-\sigma}}.
\]
Given the cost function in equation \eqref{eq:Cost=000020fcn} and
with iceberg trade costs, profit maximizing monopolistically competitive
firms producing in region $r$ set a unique f.o.b. price.}
\begin{equation}
p^{r}=1,\qquad\forall r\in\left\{ 1,2\right\} .\label{eq:p^r}
\end{equation}

Due to trade costs, the location and pricing decisions of firms affect
the price index for the composite manufacturing good, $P_{M}^{r}$.
Workers and atomistic entrepreneurs take as given the location of
firms as characterized by the density of entrepreneurs in each region
$\left\{ \lambda^{r},1-\lambda^{r}\right\} $. This results in a region-specific
manufacturing price-index, $P_{M}^{r}$, which is a generalized weighted
average of consumer prices:

\begin{equation}
\left(P_{M}^{r}\right)^{1-\sigma}=\lambda^{r}+\left(1-\lambda^{r}\right)\tau^{1-\sigma},\label{eq:P_M^s}
\end{equation}
where we use the normalization that there is no trade costs for shipments
to the own region.

Individual consumer prices are equal to producer prices for the domestically
produced varieties and to $\tau$ times the producer price for varieties
that are imported from region $r\neq s$.

\paragraph{Equilibrium Earnings, Sales, and Expenditures.}

Let $R^{r}\equiv p^{r}x$ and $\Pi^{r}\equiv R^{r}-c^{r}$ respectively
denote revenue and profits of a representative firm in region $r$
producing output $x$. Using $\mu=\frac{\alpha\sigma}{\sigma-1}$,
as well as equations \eqref{eq:Cost=000020fcn} and \eqref{eq:p^r},
we obtain $\Pi^{r}=-\pi^{r}+\frac{R^{r}}{\sigma}$. By free entry,
firm owners compete for entrepreneurs until profits are exhausted
($\Pi^{r}=0$), so that in equilibrium the individual earnings of
an entrepreneur $\pi^{r}$ are a constant fraction $\frac{1}{\sigma}$
of the revenues of the firm they manage. The rest of the revenue of
the firm accrues to workers. If all firms are of the same size in
equilibrium (which will be the case), then $w^{r}L_{M}^{r}=\frac{\sigma-1}{\sigma}R^{r}$.
Summing up, 
\begin{equation}
\pi^{r}=\frac{1}{\sigma}R^{r},\qquad w^{r}L_{M}^{r}=\left(1-\frac{1}{\sigma}\right)R^{r}.\label{eq:Revenue}
\end{equation}

The combination of monopolistic competition, iso-elastic demand curves,
and iceberg trade costs yields the following expression for sales
of a firm:\footnote{The sales of an arbitrary firm established in region $r$ to a consumer
located in $s$ spending an amount $e_{M}$ on manufacturing goods
are equal to
\[
p^{r}\tau^{rs}d^{rs}=\frac{\left(\tau^{rs}\right)^{1-\sigma}}{\lambda^{s}+\left(1-\lambda^{s}\right)\tau^{1-\sigma}}e_{M}
\]
(recall that $\tau^{rs}d^{rs}$ units must be shipped for $d^{rs}$
units to arrive in destination). Aggregating demand from all consumers
yields the following revenue for a firm established in region $r$:
\[
R^{r}=\left(\lambda^{r}m^{r}\left(u\right)\pi^{r}+\frac{\alpha}{2}\right)\frac{\left(\tau^{rr}\right)^{1-\sigma}}{\lambda^{r}+\left(1-\lambda^{r}\right)\tau^{1-\sigma}}+\left(\lambda^{s}m^{s}\left(u\right)\pi^{s}+\frac{\alpha}{2}\right)\frac{\left(\tau^{rs}\right)^{1-\sigma}}{\lambda^{s}+\left(1-\lambda^{s}\right)\tau^{1-\sigma}},\qquad s\neq r.
\]
Above, $\lambda^{s}m^{s}\left(u\right)\pi^{s}$ is the expenditure
on manufacturing from entrepreneurs located in region $s$, and $\frac{\alpha}{2}$
is the expenditure on manufacturing from workers located in either
region. Using these expressions, as well as $p^{r}=p^{s}=1$ from
equation \eqref{eq:p^r}, $\tau^{ss}=1$, $\tau^{rs}=\tau>1$ whenever
$r\neq s$, and equation \eqref{eq:P_M^s} yield the expression in
equation \eqref{eq:sales}.}
\begin{equation}
\sigma\pi^{r}=\frac{\lambda^{r}m\left(u^{r}\right)\pi^{r}+\frac{\alpha}{2}}{\lambda^{r}+\left(1-\lambda^{r}\right)\tau^{1-\sigma}}+\tau^{1-\sigma}\frac{\left(1-\lambda^{r}\right)m\left(u^{s}\right)\pi^{s}+\frac{\alpha}{2}}{\left(1-\lambda^{r}\right)+\lambda^{r}\tau^{1-\sigma}},\qquad r\neq s,\label{eq:sales}
\end{equation}
where we have used the fact that $R^{r}=\sigma\pi^{r}$ in equilibrium.
That is, the revenue of a representative firm located in region $r$
is equal to the sum of its domestic sales and of its exports. The
term $\frac{\alpha}{2}$ above corresponds to the demand from workers,
and it follows from equation \eqref{eq:w^r=00003Dalpha}.

Combining the budget constraint of an entrepreneur, $\pi^{s}=E^{s}$,
with the equilibrium value of $(P_{A},P_{M})$ from equation \eqref{eq:P_M^s}
into the expenditure function \eqref{eq:Expenditure=000020Function=000020E-2},
and using the definition of the real earnings of entrepreneurs \eqref{eq:Omega},
yields the following equilibrium expression for location $r$:
\begin{equation}
\ln\pi^{r}=\ln\Omega(u^{r})+m\left(u^{r}\right)\ln P_{M}^{r}.\label{eq:Expenditure=000020=00003D=000020income}
\end{equation}

We can now return to assumption \ref{assu:omega} and lay out sufficient
conditions on the parameters of the model that ensure this assumption
holds.
\begin{lem}
\label{lem:Sufficient=000020condition}\textbf{Sufficient regularity
conditions}. If parameter values are such that the following inequality
holds:
\begin{equation}
\frac{\gamma\sigma\tau}{\text{\ensuremath{\underline{s}}}}<1,\label{eq:sufficient=000020condition=000020--=000020fundamentals}
\end{equation}
then the sufficient condition in assumption \ref{assu:omega} holds.
\end{lem}

Note that this condition is satisfied as long as trade costs $\tau$
are finite since we let $\gamma\rightarrow0$.

\paragraph{Sectoral Employment.}

In equilibrium, the value added of each factor equals its contribution
to the generation of revenue. If real and nominal earnings of entrepreneurs
are equalized across regions, then we can write $\sum_{r}w^{r}L_{A}^{r}=\frac{1-m}{\sigma}\sum_{r}R^{r}$
(the value of labor in agriculture is equal to expenditure in agriculture,
where $1-m$ is the expenditure share of entrepreneurs in agriculture)
and $\sum_{r}w^{r}L_{M}^{r}=\frac{\sigma-1}{\sigma}\sum_{r}R^{r}$
(the value of labor in manufacturing is a fraction $\frac{\sigma-1}{\sigma}$
of sales). Finally, full-employment of labor implies $L_{A}+L_{M}=1$,
where $L_{i}\equiv\sum_{r}L_{i}^{r}$ for $i=A,M$. Combining these
expressions yields the following equilibrium allocation of workers
across sectors
\begin{equation}
\frac{L_{M}}{1-L_{M}}=\frac{\sigma-1}{1-m}.\label{eq:employment}
\end{equation}

Equation \eqref{eq:employment} implies that the employment share
in manufacturing depends on the level of development as captured by
how the desired consumption share of manufacturing $m(u)$ depends
on the level of utility, $u.$ We obtain the following result:
\begin{prop}
\label{prop:Structural=000020change}\textbf{Engel curves and structural
change}. The allocation of workers across sectors follows the expenditure
shares. Given assumption \ref{assu:Engel}, which implies that manufactures
are a luxury good (and the agricultural good is a necessity), it follows
that, $L_{M}$ is increasing in the level of utility $u$ in equilibrium.
\end{prop}

This proposition implies that, as sustained technological progress
raises living standards over time, consumption and employment gradually
shift away from agriculture and the economy industrializes. Is this
structural change even across space? Does it lead to regional disparities?
And is there a feedback effect of equilibrium location choices on
structural change? The next sections address these questions. We first
characterize the equilibrium conditional on a level of technological
progress (labor productivity $\alpha$ in our model) and then analyze
the effect of rising productivity.

\section{\protect\label{sec:Equilibrium}Spatial Equilibriums}

The previous section only provided a partial characterization of the
equilibrium by taking the location of the mobile skilled workers (entrepreneurs)
$\left\{ \lambda^{r},1-\lambda^{r}\right\} $ and their utility levels
$\left\{ u^{r},u^{s}\right\} $ as given. Here we finalize the characterization
of the equilibrium by solving for the spatial allocation of entrepreneurs.
Several expressions in this section are more compact if we evaluate
utility using real earnings $\Omega$ instead of $u$ -- recall from
equation \eqref{eq:Omega} that $\Omega$ is increasing in $u$ and
thus also in indirect utility $v$.

Entrepreneurs are freely mobile. They locate where they obtain their
highest utility, and the spatial allocation of firms is an equilibrium
outcome. It is useful to think about the spatial equilibrium as a
complementarity slackness condition. We say that a spatial equilibrium
arises whenever real earnings obey $\Omega^{r}\leq\Omega$, for some
$\Omega>0$, for all $r\in\left\{ 1,2\right\} $, with equality if
$\lambda^{r}>0$. Following \citet{Krugman1991}, when $\lambda^{r}=1$,
we label region $r$ the ``core'' and region $s\neq r$ the ``periphery.''
We focus our analysis on the two potentially stable equilibriums that
can arise: the symmetric equilibrium $\lambda^{r}=\frac{1}{2}$ for
$r\in\left\{ 1,2\right\} $, in which both regions are identical,
and the so-called ``core-periphery'' equilibrium, in which $\lambda^{r}=1$
for some $r\in\left\{ 1,2\right\} $ and regions with identical fundamentals
diverge in equilibrium (region $r$ hosts all manufactures and region
$s\neq r$ fully specializes in agriculture). Note that in each of
these two equilibriums, all entrepreneurs enjoy the same level of
utility, and hence spend the same share $m\left(\Omega\right)$ on
manufactures -- but the level of utility itself typically differs
\emph{across} equilibrium configurations.

A core-periphery equilibrium with all entrepreneurs agglomerated in
a single region, i.e., $\lambda^{r}=1$ or $\lambda^{r}=0$, may exist.
The symmetric configuration $\lambda^{r}=\frac{1}{2}$ is always an
equilibrium of this symmetric model. It may not always be stable,
though. We start by characterizing the core-periphery equilibrium
and the condition under which it exists. We then characterize the
conditions under which the symmetric equilibrium is stable.

\subsection{The Core-Periphery Equilibrium}

\paragraph*{Necessary Condition: The Sustain Point.}

What are the conditions that ensure $\lambda^{r}\in\left\{ 0,1\right\} $
is an equilibrium outcome? Without loss of generality, we proceed
by assuming $\lambda^{r}=1$, i.e., $r$ is the core and $s$ the
periphery (we think about ``$s$'' as a mnemonic for ``small''
region). In this case, $P_{M}^{r}=1$ by equation \eqref{eq:P_M^s}
and $\pi^{r}=\frac{\alpha}{\sigma-m}$ by equation \eqref{eq:sales}.
Using equation \eqref{eq:Expenditure=000020=00003D=000020income}
yields the level of real earnings $\Omega^{C}$of entrepreneurs operating
from the core, 
\begin{equation}
\Omega^{C}=\pi^{C}=\frac{\alpha}{\sigma-m\left(\Omega^{C}\right)},\label{eq:v=000020Core}
\end{equation}
where the superscript ``$C$'' refers to ``Core-Periphery'' equilibrium
values. Equation \eqref{eq:v=000020Core} gives us an equilibrium
relationship between real earning $\Omega^{C}$ and parameters $\left\{ \alpha,\sigma,\tau\right\} $.
Note that all entrepreneurs access all manufacturing varieties at
f.o.b. prices, hence $\tau=1$ in the triplet above, and their real
income $\Omega^{C}$ is equal to their nominal income $\pi^{C}$.
In turn, by equation \eqref{eq:expenditure=000020share=000020range},
we also obtain an equilibrium relationship between the expenditure
share on manufactures, $m^{C}$, and parameters $\left\{ \alpha,\sigma,\tau\right\} $,
where $\tau=1$. We impose the following regularity condition to ensure
that this equilibrium relationship is not degenerate, in which case
$\Omega^{C}$ and $m^{C}$ are functions that are increasing in labor
productivity $\alpha$ and decreasing in product differentiation $\sigma$:
\begin{assumption}
\label{assu:regularity-1}Assume the functions $\omega_{A}$ and $\omega_{M}$
and parameter $\sigma$ are such that 
\[
\frac{\partial\ln m}{\partial\ln\Omega}<\sigma-1.
\]
\end{assumption}

This assumption, which requires that income effects are bounded above,
ensures that equation \eqref{eq:v=000020Core} admits a unique fixed
point, and that real earnings $\Omega^{C}$ and manufacturing expenditure
share $m^{C}\equiv m\left(\Omega^{C}\right)=m^{C}\left(\alpha,\sigma,1\right)$
are increasing in productivity $\alpha$ and increasing in product
differentiation (decreasing in $\sigma$):
\begin{equation}
\frac{\partial\ln m}{\partial\ln\Omega}<\sigma-1\qquad\Rightarrow\qquad\frac{\partial m^{C}}{\partial\sigma},\frac{\partial\Omega^{C}}{\partial\sigma}<0<\frac{\partial m^{C}}{\partial\alpha},\frac{\partial\Omega^{C}}{\partial\alpha}.\label{eq:CP=000020alpha=000020sigma}
\end{equation}

Consider now an atomistic entrepreneur entertaining the possibility
of supplying their variety from the periphery instead (region $s$).
Using equation \eqref{eq:sales}, their (shadow) income $\pi^{s}$
would be equal to
\begin{align}
\pi^{s} & =\frac{\tau^{1-\sigma}}{\sigma}\left[m\left(\Omega^{s}\right)\frac{\alpha}{\sigma-m\left(\Omega^{s}\right)}+\frac{\alpha}{2}\left(1+\frac{1}{\tau^{2(1-\sigma)}}\right)\right],\label{eq:pi^s}
\end{align}
and they would face consumer price $P_{M}^{s}=\tau$. Hence, their
potential real income $\Omega^{s}$ would be:
\begin{equation}
\ln\Omega^{s}=\ln\pi^{s}-m\left(\Omega^{s}\right)\ln\tau.\label{eq:v=000020periphery-1}
\end{equation}
Such a move would leave them at best equally well-off if and only
if $\Omega^{s}\leq\Omega^{C}$, in which case the core-periphery outcome
is said to be ``sustainable.'' Using equations \eqref{eq:v=000020Core}
and \eqref{eq:v=000020periphery-1}, we find that $\Omega^{s}=\Omega^{C}$
if and only if
\begin{equation}
1-\frac{2\sigma}{\sigma-m^{C}\left(\alpha,\sigma,1\right)}\tau^{m^{C}\left(\alpha,\sigma,1\right)+1-\sigma}+\frac{\sigma+m^{C}\left(\alpha,\sigma,1\right)}{\sigma-m^{C}\left(\alpha,\sigma,1\right)}\tau^{2\left(1-\sigma\right)}=0.\label{eq:sustain=000020plane}
\end{equation}
Observe that $\tau=1$ is always a root of equation \eqref{eq:sustain=000020plane},
regardless of $\left(\alpha,\sigma\right)$.\footnote{The economic meaning of this results is that location does not matter
in the absence of trade costs.} By Descartes' generalization of Laguerre's rule of signs, this expression
does not admit any solution for $\tau$ in $\left(1,\infty\right)$
if $m^{C}\left(\alpha,\sigma,1\right)\geq\sigma-1$. Conversely, if
$m^{C}\left(\alpha,\sigma,1\right)<\sigma-1$, equation \eqref{eq:sustain=000020plane}
defines a hyperplane for $\left(\alpha,\sigma,\tau\right)\in\left(1,\infty\right)^{3}$,
and in particular it admits a root for $\tau>1$, which we denote
by $\tau_{1}$. Standard algebra reveals that $\tau_{1}$ is increasing
in $\alpha$. In turn, the core-periphery outcome is an equilibrium
($\Omega^{s}\leq\Omega^{C}$) only if trade costs are sufficiently
low and productivity is sufficiently high. Formally, we establish
the following proposition in the appendix:
\begin{lem}
\label{lem:Sustain=000020point-2}\textbf{Sustain point}. Let $\alpha_{1}\left(\sigma\right)>0$
be defined implicitly by $m\left(\alpha_{1},\sigma,1\right)=\sigma-1$
when $\sigma\in\left(1,2\right]$, with $\alpha_{1}^{\prime}>0$ by
equation \eqref{eq:CP=000020alpha=000020sigma}. (i) If $\sigma\in\left(1,2\right]$
and $0<\alpha\leq\alpha_{1}\left(\sigma\right)$, then $\Omega^{s}\leq\Omega^{C}$
if $\lambda^{r}=1$ for any value of $\tau$ in $\left[1,\tau_{1}\right]$.
(ii) If $\sigma\in\left(1,2\right]$ and $\alpha>\alpha_{1}\left(\sigma\right)$,
then $\Omega^{s}\leq\Omega^{C}$ if $\lambda^{r}=1$ for any $\tau\geq1$.
(iii) If $\sigma>2$, then $\Omega^{s}\leq\Omega^{C}$ if $\lambda^{r}=1$
for any value of $\tau$ in $\left[1,\tau_{1}\right]$.
\end{lem}

In words, we have that for low values of elasticity of substitution,
$\sigma\in(1,2],$ when productivity is low, $0<\alpha\leq\alpha_{1}\left(\sigma\right)$,
the core-periphery equilibrium equilibrium exists only for transportation
costs below $\tau_{1}\left(\alpha\right)$. Instead, when labor productivity
is high, $\alpha>\alpha_{1}\left(\sigma\right),$ the core-periphery
equilibrium exists regardless of the level of transportation costs.
By contrast, for high values of the elasticity of substitution $\sigma>2,$
the equilibrium exists for transportation costs below $\tau_{1}$
regardless of the level of labor productivity. We next turn to the
so-called ``break point'' which provides sufficient conditions for
the existence of a core-periphery equilibrium.

\subsection{The Symmetric Equilibrium}

A symmetric equilibrium with $\lambda^{r}=\frac{1}{2}$ always exists
in this model -- this can be shown by direct inspection of the equilibrium
conditions. In this case, prices and earnings are equalized across
regions, with $\left(P_{M}\right)^{1-\sigma}=\frac{1+\tau^{1-\sigma}}{2}$
by equation \eqref{eq:P_M^s} and $\pi=\frac{\alpha}{\sigma-m}$ by
equation \eqref{eq:sales}. In turn, using indirect utility equation
\eqref{eq:Indirect=000020Utility=000020--=000020intermediate=000020step},
the equilibrium real income level $\Omega^{B}$ obeys:
\begin{equation}
\ln\Omega^{B}=\ln\frac{\alpha}{\sigma-m\left(\Omega^{B}\right)}-\frac{m\left(\Omega^{B}\right)}{\sigma-1}\ln\left(\frac{2}{1+\tau^{1-\sigma}}\right).\label{eq:v=000020symmetric-1}
\end{equation}
This expression admits at most one interior fixed point for real earnings
$\Omega^{B}$ as a function of parameters $\left(\alpha,\sigma,\tau\right)$.\footnote{The derivative of the left-hand side of expression \eqref{eq:v=000020symmetric-1}
with respect to $\ln\Omega^{B}$ is equal to one. The derivative of
its right-hand side with respect to $m\left(\Omega^{B}\right)$ is
smaller than $\frac{1}{\sigma-1}$. Hence, under assumption \ref{assu:regularity-1},
the derivative of the right-hand side with respect to $\ln\Omega^{B}$
is smaller than $m$, which is itself smaller than one. Thus, expression
\eqref{eq:v=000020symmetric-1} admits at most one interior fixed
point for $\Omega^{B}$.}

\paragraph*{Sufficient Condition: The Break Point}

This symmetric equilibrium is said to be unstable if, following an
exogenous migration shock $\hat{\lambda}$, the change in income,
$\hat{\pi}$, is larger than the change in expenditure, $\hat{E}$,
required to maintain utility, and to be stable otherwise (here we
use ``hats'' to denote log changes). We are interested in the parameter
configuration such that the two effects are exactly equal at the margin,
$\hat{\pi}=\hat{E}=m^{B}\hat{P}_{M}$, in which case utility (and
hence expenditure shares) are invariant. This condition admits at
most one root for $\tau\neq1$, and this root is the solution to:\footnote{As was the case for sufficient conditions for a core-periphery equilibrium,
$\tau=1$ is a root for any combination of $\alpha$ and $\sigma$,
as location does not matter in the absence of trade costs.}
\begin{equation}
\tau^{\mathrm{1-\sigma}}=\frac{\sigma-1-m^{B}\left(\alpha,\sigma,\tau\right)}{\sigma-1+m^{B}\left(\alpha,\sigma,\tau\right)}\frac{\sigma-m^{B}\left(\alpha,\sigma,\tau\right)}{\sigma+m^{B}\left(\alpha,\sigma,\tau\right)}.\label{eq:break=000020point}
\end{equation}
This expression does not admit any root for $\tau$ in $\left(1,\infty\right)$
if $m^{B}\left(\alpha,\sigma,\infty\right)\geq\sigma-1$. Conversely,
if $m^{B}\left(\alpha,\sigma,\infty\right)<\sigma-1$, equation \eqref{eq:sustain=000020plane}
defines a hyperplane for $\left(\alpha,\sigma,\tau\right)\in\left(1,\infty\right)^{3}$,
and in particular it admits a root for $\tau>1$, which we denote
by $\tau_{0}$.\footnote{The left-hand side of this expression decreases from one when $\tau=1$
to zero when $\tau\rightarrow\infty$. Its right-hand side increases
from a number smaller than one (possibly negative) when $\tau=1$
to a larger number when $\tau\rightarrow\infty$ (recall that $m^{B}$
is decreasing in $\tau$). This larger number belongs to the interior
of the unit interval if and only if $m^{B}\left(\alpha,\sigma,\infty\right)<\sigma-1$.
Thus, this expression admits a unique solution for $\tau\in\left(1,\infty\right)$
if and only if $m^{B}\left(\alpha,\sigma,\infty\right)<\sigma-1$,
as was to be shown.} Standard algebra reveals that $\tau_{0}$ is increasing in $\alpha$.
In turn, the symmetric equilibrium is unstable ($\hat{\pi}\geq m^{B}\hat{P}_{M}$)
if trade costs are sufficiently low and productivity is sufficiently
high.

We can establish the following result:
\begin{lem}
\label{lem:Break=000020point-2}\textbf{Break point}. Let $\alpha_{\infty}\left(\sigma\right)>0$
be defined implicitly by $m\left(\alpha_{\infty},\sigma,\infty\right)=\sigma-1$
when $\sigma\in\left(1,2\right]$, with $\alpha_{\infty}^{\prime}>0$
by equation \eqref{eq:CP=000020alpha=000020sigma}. (i) If $\sigma\in\left(1,2\right]$
and $0<\alpha<\alpha_{\infty}\left(\sigma\right)$, then, the symmetric
equilibrium is unstable for any value of $\tau$ in $\left[1,\tau_{0}\right]$.
(ii) If $\sigma\in\left(1,2\right]$ and $\alpha>\alpha_{\infty}\left(\sigma\right)$,
then the symmetric equilibrium is unstable for any $\tau\geq1$. (iii)
If $\sigma>2$, then the symmetric equilibrium is unstable for any
value of $\tau$ in $\left[1,\tau_{0}\right]$.
\end{lem}

\subsection{Equilibrium Configurations}

The symmetric and core-periphery equilibrium configurations share
similarities. In both equilibriums, given the expenditure share $m$,
the nominal earnings of each type of agent are related to parameters
$\left\{ \alpha,\sigma,\tau\right\} $ and they are determined by
the same equations, $w=\alpha$ and $\pi=\frac{\alpha}{\sigma-m}$.
The key difference is that the equilibrium expenditure shares and
levels of utility differ. Using equations \eqref{eq:v=000020periphery-1}
and \eqref{eq:v=000020symmetric-1}, we obtain the following result.
\begin{lem}
\label{lem:break=000020sustain=000020points=000020ranking}\textbf{Equilibriums
ranking}. (i) Given parameter values $\left\{ \alpha,\sigma,\tau\right\} $,
the real earnings of entrepreneurs are strictly higher at the core-periphery
equilibrium than at the symmetric equilibrium:
\begin{equation}
\forall\left\{ \alpha,\sigma,\tau\right\} :\quad\Omega^{B}\left(\alpha,\sigma,\tau\right)<\Omega^{C}\left(\alpha,\sigma,1\right).\label{eq:utility=000020ranking}
\end{equation}
This result implies in turn $m^{B}<m^{C}$, and $\pi^{B}<\pi^{C}$.
(ii) Given parameter values $\left\{ \alpha,\sigma\right\} $, $\tau_{1}>\tau_{0}$.
(iii) Any spatial equilibrium other than $\lambda\in\left\{ 0,\frac{1}{2},1\right\} $,
if it exists, is not stable.
\end{lem}

Part (iii) of this lemma allows us to study only two potential equilibrium
configurations -- the core-periphery pattern, and the symmetric equilibrium.
We are now in position to establish the main results of this section.
\begin{prop}
\label{prop:Break=000020and=000020Sustain=000020points}If $\sigma\in\left(1,2\right]$,
then $\exists\alpha_{1}\left(\sigma\right),\alpha_{\infty}\left(\sigma\right)>0$
defined implicitly as $m^{C}\left(\alpha_{1},\sigma,1\right)=\sigma-1$
and $m^{B}\left(\alpha_{\infty},\sigma,\infty\right)=\sigma-1$, respectively,
with $0<\alpha_{1}\left(\sigma\right)<\alpha_{\infty}\left(\sigma\right)$.
(i) The symmetric equilibrium is the unique stable equilibrium configuration
in either of the following cases: if $\sigma>2$ and $\tau>\tau_{1}\left(\alpha,\sigma\right)$;
or if $\sigma\in\left(1,2\right]$, $0<\alpha<\alpha_{1}\left(\sigma\right)$,
and $\tau>\tau_{1}\left(\alpha,\sigma\right)$. (ii) The core-periphery
configuration is the unique stable equilibrium in any of the following
cases: if $\sigma>2$ and $\tau\leq\tau_{0}\left(\alpha,\sigma\right)$;
if $\sigma\in\left(1,2\right]$, $0<\alpha<\alpha_{\infty}\left(\sigma\right)$,
and $\tau\leq\tau_{0}\left(\alpha,\sigma\right)$; if $\sigma\in\left(1,2\right]$
and $\alpha\geq\alpha_{\infty}$$\left(\sigma\right)$. (iii) The
symmetric equilibrium and the core-periphery pattern are both stable
equilibrium configurations in any of the following cases: if $\sigma>2$
and $\tau_{0}\left(\alpha,\sigma\right)<\tau\leq\tau_{1}\left(\alpha,\sigma\right)$;
or if $\sigma\in\left(1,2\right]$, $0<\alpha<\alpha_{1}\left(\sigma\right)$,
and $\tau_{0}\left(\alpha,\sigma\right)<\tau\leq\tau_{1}\left(\alpha,\sigma\right)$;
or if $\sigma\in\left(1,2\right]$, $\alpha_{1}\left(\sigma\right)<\alpha<\alpha_{\infty}\left(\sigma\right)$,
and $\tau>\tau_{0}\left(\alpha,\sigma\right)$.
\end{prop}

Figure \ref{fig:equil} graphically depicts the type of equilibriums
that we can obtain depending on the level of productivity, $\alpha,$
and transportation costs, $\tau$ for low and high values of $\sigma,$
respectively $\sigma\in(1,2]$ and $\sigma>2.$ We return to these
alternative equilibrium configurations in the next section, where
we consider an economy in which long-run growth is driven by labor
productivity growth.

\begin{figure}
\caption{Equilibrium Configurations}
\label{fig:equil}\subfloat[{Case $\sigma\in(1,2]$}]{

\includegraphics[scale=0.4]{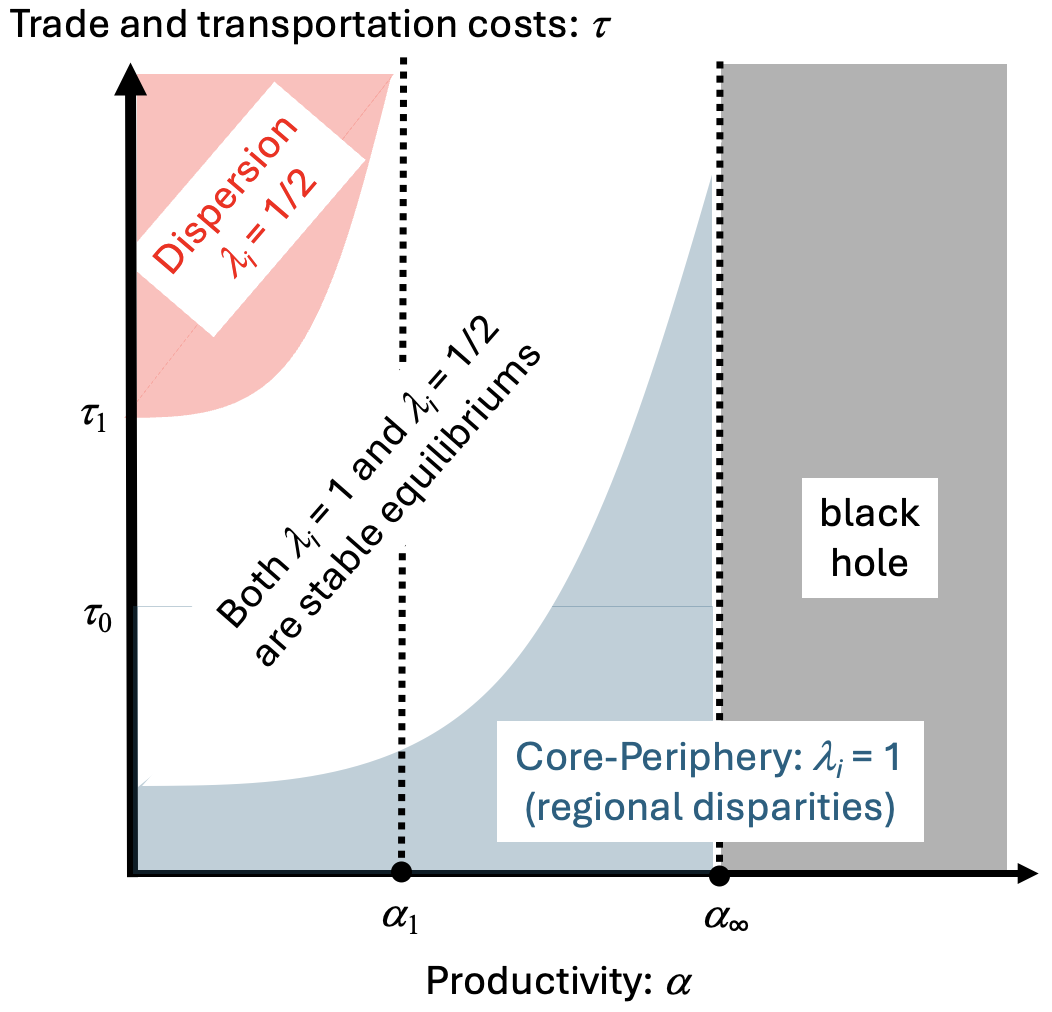}}\subfloat[Case $\sigma>2$]{

\includegraphics[scale=0.4]{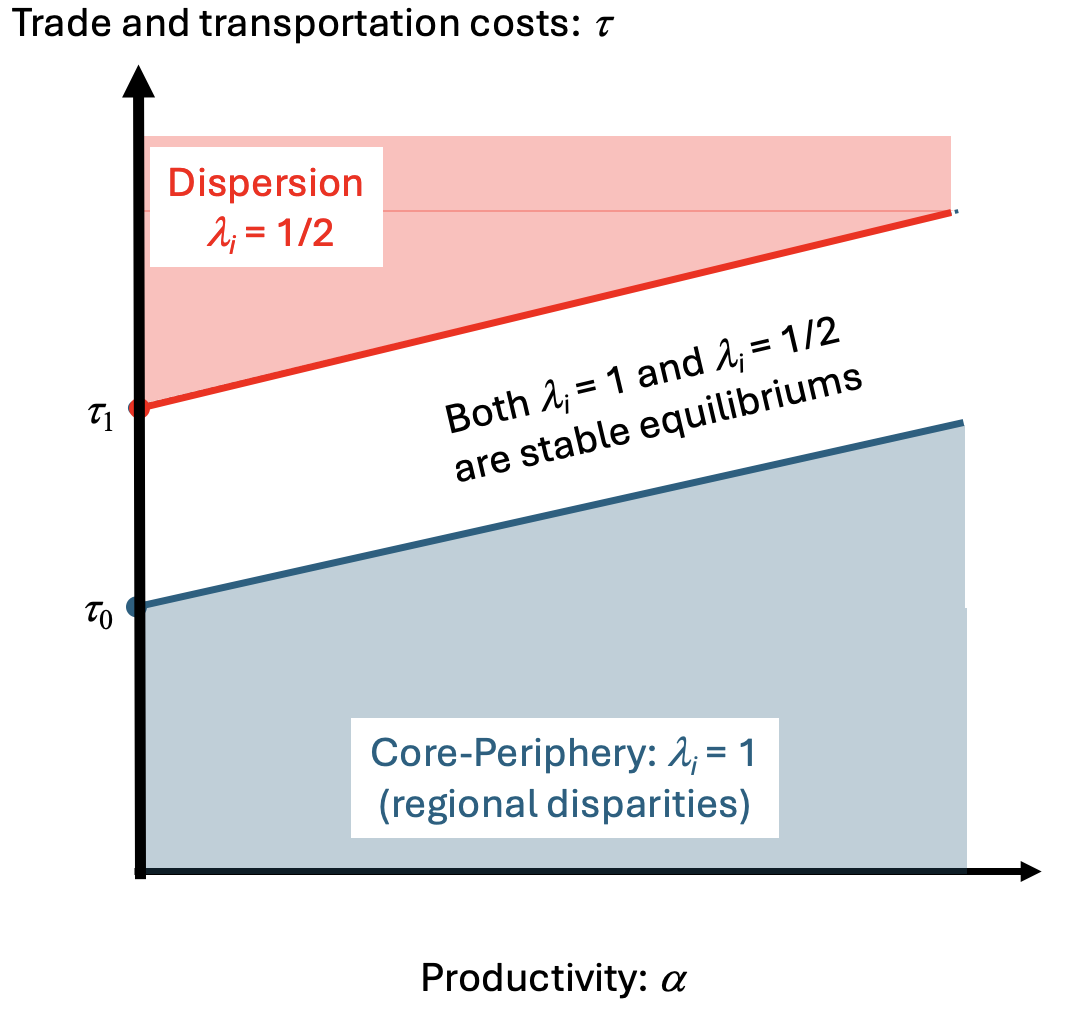}}
\end{figure}

\paragraph{The Emergence of Black Holes.}

Note that the core-periphery configuration emerges as the unique stable
equilibrium for any level of trade costs $\tau>1$ if $\sigma\in\left(1,2\right]$
and $\alpha>\alpha_{\infty}\left(\sigma\right)$ because, in this
case, agglomeration forces, as captured by the share of mobile expenditure
$m^{B}$, always dominate dispersion forces, as captured by $\sigma-1$,
and hence $m^{B}\geq\sigma-1$. \citet{Krugman1991} refers to this
case as a ``black hole.'' In his model, $m$ is a preference parameter
and he rules out this case by imposing his famous ``no black-hole
condition'', as in this case his model fails to feature the richness
of equilibrium configurations that the complementary parameter space
provides. In our model, $m$ is an endogenous variable and the black
hole configuration emerges as an equilibrium outcome when labor productivity
is high enough. An important implication of this results is that prohibitive
transportation costs in a sector are not inconsistent with its spatial
clustering. Our model can thus accommodate the agglomeration of non-traded
services as an equilibrium outcome.

\paragraph*{Robustness to the Helpman Critique.}

The qualitative results in this section are robust to Helpman's critique
of Krugman's model \citep{helpman1998size}. In Krugman's original
model, the source of immobile demand is agriculture. This good produces
under constant returns to labor and it is freely traded. Its price
depends only on model parameters in equilibrium. \citet{helpman1998size}
points out the fragility of the result by reversing the model assumptions.
He introduces housing as the source of immobile demand. Housing is
among the most inherently non-tradable of services and its equilibrium
return depends on the spatial distribution of economic activity. As
a result, Helpman shows that in this alternative setting dispersion
forces dominate agglomeration forces when inter-regional trade and
transportation costs in manufacturing are low enough, turning this
specific result of Krugman's on its head. We show formally in Appendix
\ref{Appx:Helpman} that our results are robust to Helpman's criticism.
We develop an extension of the model with a non-traded factor --
structures -- which yields a dispersion force that is independent
of transportation costs in manufacturing. Thus, the symmetric equilibrium
is stable for arbitrarily small levels of transportation costs in
this alternative setting. Nevertheless, we show that the symmetric
equilibrium is unstable for arbitrarily large levels of trade costs
if $m>\sigma-1$, which is the exact same condition as in our current
model.

\section{\protect\label{sec:Structural=000020Change=000020Spatial=000020Disparities}Structural
Change and Spatial Disparities}

We are now ready to close the analysis of the interplay between structural
change and the emergence of spatial disparities. Building on the tradition
of the neoclassical growth model, we consider an economy in which
growth is driven by exogenous labor-augmenting technical progress.
As in \citet*{Kongsamut2001}, this uniform productivity growth drives
the reallocation of labor from agriculture to manufacturing due to
non-homothetic demand. For simplicity, there is no capital accumulation
in our model, and different periods of time are a succession of static
equilibriums. Given the spatial equilibrium configuration $j\in\left\{ B,C\right\} $,
the national (nominal) income is equal to 
\begin{equation}
Y^{j}\left(\alpha,\sigma,\tau\right)\equiv w+\pi^{j}\left(\alpha,\sigma,\tau\right)=\alpha\left(1+\frac{1}{\sigma-m\left(\alpha,\sigma,\tau\right)}\right).\label{eq:National=000020income}
\end{equation}
Per capita nominal income is equal to $\frac{Y}{2}$.The share of
manufacturing in national income is equal to
\begin{equation}
\tilde{m}^{j}\left(\alpha,\sigma,\tau\right)\equiv\frac{\pi^{j}\left(\alpha,\sigma,\tau\right)}{\pi^{j}\left(\alpha,\sigma,\tau\right)+w}=\frac{1}{1+\sigma-m^{j}\left(\alpha,\sigma,\tau\right)},\label{eq:m=000020share=000020income}
\end{equation}
and the earnings premium of entrepreneurs is equal to: 
\begin{equation}
\tilde{\pi}^{j}\left(v\right)\equiv\frac{\pi^{j}\left(v\right)}{w}=\frac{1}{\sigma-m^{j}\left(\alpha,\sigma,\tau\right)}.\label{eq:premium}
\end{equation}
By inspection, all these variables are increasing in the manufacturing
expenditure share $m$.\footnote{Specifically, as $m$ increases from zero to one, $Y$ increases from
$\alpha\frac{1+\sigma}{\sigma}$ to $\alpha\frac{\sigma}{\sigma-1}$
, the share $\tilde{m}$ increases from $\frac{1}{1+\sigma}$ to $\frac{1}{\sigma}$
, and the earnings premium $\tilde{\pi}$ increases from $\frac{1}{\sigma}$
to $\frac{1}{\sigma-1}$. Note that this ``premium'' may be smaller
than unity by our choice of units for labor and human capital. It
is straightforward to make a different choice of units that would
ensure that the premium is at larger than $1$, but at the cost of
carrying an additional parameter throughout. The important point for
our purposes is that this premium is increasing in $m$.} Thus, labor-augmenting technical progress -- modeled as a continuous
increase of $\alpha$ over time -- raises national and per-capita
incomes (by equation \ref{eq:National=000020income}), the share of
manufacturing in the economy (by equation \ref{eq:m=000020share=000020income}),
and the earnings premium of entrepreneurs (by equation \ref{eq:premium}).
Formally:
\begin{prop}
\label{prop:Growth=000020and=000020Structural=000020change}\textbf{Growth
and structural change}. Assume labor productivity $\alpha$ increases
over time. Then, given any spatial equilibrium $j\in\left\{ B,C\right\} $,
national income $Y$, per-capita income $\frac{Y}{2}$, the share
of manufacturing $\tilde{m}$, and the earnings premium of entrepreneurs
$\tilde{\pi}$ all increase over time.
\end{prop}

Observe that the effect of productivity growth on structural change
and earnings inequality is exclusively mediated through the manufacturing
expenditure share $m$. That is, $m$ is increasing in labor productivity
$\alpha$ because manufactures are a luxury good. By contrast, labor
productivity has both a direct effect on income and an indirect effect
through $m$. As we have extensively discussed in the previous section,
technical progress also drives the spatial reallocation of manufacturing
firms. To understand how our model brings together the notions of
structural change and spatial disparities, it is useful to return
to the abstract of the seminal \citet{Krugman1991} paper:
\begin{quotation}
This paper develops a simple model that shows how a country can endogenously
become differentiated into an industrialized ``core'' and an agricultural
``periphery.'' In order to realize scale economies while minimizing
transport costs, manufacturing firms tend to locate in the region
with larger demand, but the location of demand itself depends on the
distribution of manufacturing. Emergence of a core-periphery pattern
depends on transportation costs, economies of scale, and the share
of manufacturing in national income {[}p.483{]}.
\end{quotation}
This abstract mentions three important drivers of regional disparities
in the model. The first one is transportation costs ($\tau$ in the
model). The typical comparative statics exercise of most contributions
in the \textsc{neg} literature is to consider the consequences of
a steady reduction of trade and transportation costs. This particular
focus generated a lot of interest particularly in Europe. Summarizing
this body of work, \citet{Puga2002} concludes that the rise of regional
disparities in Europe is a likely consequence of massive investments
in transportation infrastructures and of the implementation of the
Single Market Program.\footnote{\citet{Faber2014} provides direct empirical evidence of this effect
in the context of the development of the National Trunk Highway System
in China.} The second one is scale economies. In equilibrium, unexploited scale
economies are increasing in product differentiation (decreasing in
$\sigma$) by free-entry. The lower $\sigma$, the larger the equilibrium
size of firms. A higher $\sigma$ also makes demand for individual
varieties more sensitive to trade costs, as is captured by the combination
of parameters, $\tau^{1-\sigma}$, that enters the equilibrium conditions.
The final ingredient in Krugman's abstract is the relative importance
of the sector featuring scales economies -- manufacturing -- in
the national economy. The larger it is, the stronger are self-reinforcing
agglomeration economies, and the more likely regional disparities
emerge in equilibrium. This share is a $parameter$ in Krugman's core-periphery
model. It is an endogenous $outcome$ in ours.

Consider, then, the following thought experiment. Assume that labor
productivity is initially low, with $\alpha<\alpha_{1}$, and trade
costs are initially high, with $\tau>\tau_{1}$, so that manufacturing
is evenly dispersed across both regions at the unique stable equilibrium.
Assume further that labor productivity monotonically increases over
time, holding other parameters fixed (in particular $\tau$ and $\sigma$).
Per capita incomes and economic well-being rise. In turn, the share
of spending on manufacturing increases and that on agriculture falls.
The first implication of this process is that labor reallocates from
agriculture to manufacturing: \emph{structural change} takes place
(see proposition \ref{prop:Growth=000020and=000020Structural=000020change}).
As the share of spending on manufactures increases, a second implication
is that agglomeration forces strengthen and dispersion forces weaken.
Eventually, the spatial concentration of manufacturing becomes a possibility
(as $\alpha$ rises above $\alpha_{1}$), and, then, from a possibility,
it becomes the only stable equilibrium outcome (as $\alpha$ rises
above $\alpha_{\infty}$): \emph{regional disparities} emerge (see
proposition \ref{prop:Break=000020and=000020Sustain=000020points}).

This evolution is graphically depicted in Figure \ref{fig:trajectory}.
As productivity grows over time and utility increases, the expenditure
share in manufacturing rises. Eventually, the core-periphery equilibrium
becomes possible since the sustain point is crossed. As productivity
keeps growing, the core-periphery equilibrium is the only stable equilibrium
and agglomeration of manufacturing at the core is a foregone conclusion.

\begin{figure}
\caption{Dynamics of Spatial Structural Change}
\label{fig:trajectory}

\includegraphics[scale=0.45]{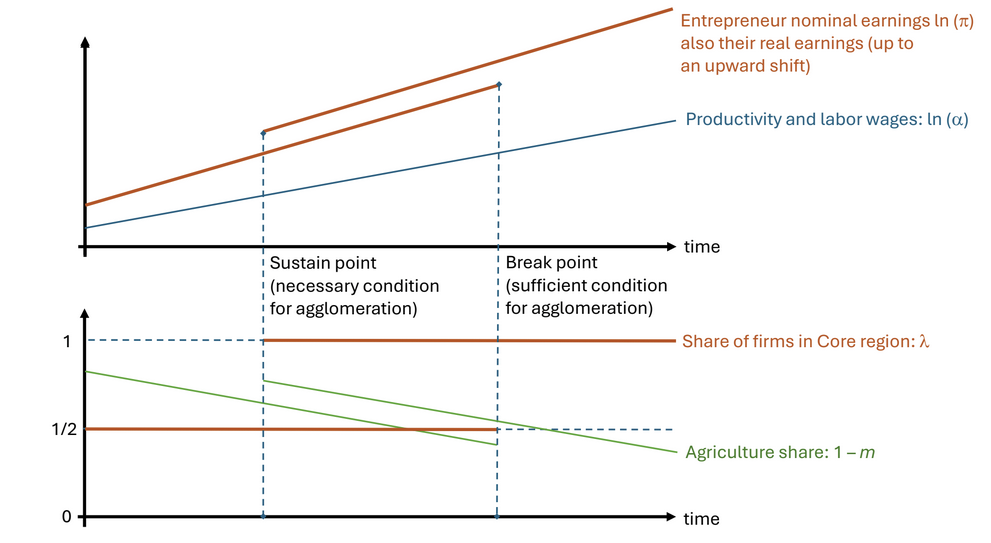}
\end{figure}

The relationship between regional disparities and structural change
goes both ways. The emergence of regional disparities has a positive
feedback effect on structural change because, as industry concentrates
in a region, the economic well being of entrepreneurs increases, everything
else equal, and with it their expenditure on manufacturing (see lemma
\ref{lem:break=000020sustain=000020points=000020ranking}). As a result,
the national expenditure share in manufacturing further increases.
This two-way interaction between agglomeration and structural change
is absent from standard structural change models that abstract from
the study of the spatial distribution of economic activity. Our model
is stylized in that there are only two locations and only two potential
stable equilibriums: symmetric or core-periphery. We believe that
extending the model to more locations would magnify the importance
of the two-way interaction we described, and we plan to study it in
subsequent work.

Structural change brought about by the combination of the steady growth
of labor productivity and of manufacturing being a luxury good also
has an effect on individual income inequality. Though it benefits
both types of factor owners, productivity growth disproportionately
benefits the owners of the factor used intensively in the sector growing
the most -- entrepreneurs (see proposition \ref{prop:Growth=000020and=000020Structural=000020change}).
The logic is well-known since at least \citet{Jones1965}, and here
it extends to the case of a factor specific to a monopolistically
competitive sector. Summarizing:
\begin{prop}
\label{prop:Structural=000020change=000020and=000020spatial=000020disparities}\textbf{Structural
change and spatial disparities}. Assume labor productivity $\alpha$
increases over time, all other things equal. Then, starting from an
initially low value $\alpha<\alpha_{1}$, (i) structural change takes
place: expenditure and employment are reallocated from agriculture
to manufacturing; (ii) regional disparities may (if $\alpha_{1}<\alpha<\alpha_{\infty}$)
and will (if $\alpha_{\infty}<\alpha$) emerge in equilibrium; (iii)
regional disparities reinforce structural change; and (iv) structural
change leads to an increase in the earnings premium of entrepreneurs.
\end{prop}

\section{Conclusion}

In this paper, we develop a parsimonious model featuring non-homothetic
demand to bring together the notions of structural change and regional
disparities. To this aim, we introduce Heterothetic Cobb-Douglas preferences,
which combine income effects with a unitary price elasticity of demand.
This latter property keeps the model extremely tractable, and we expect
it to be useful to applications in many other contexts. We incorporate
a two-region \textsc{neg} model into a two-sector structural change
framework. As labor productivity exogenously and monotonically increases
over time, economic well-being rises and the expenditure share on
agricultural goods falls. As a result, labor reallocates away from
agriculture, and industry concentrates spatially: structural change
and regional disparities are two outcomes of the growth process. Moreover,
the industrial concentration increases productivity, which in turn
leads to higher real incomes and a further reallocation away from
agriculture. As such, structural change and regional disparities are
also two mutually reinforcing propagators of the growth process.

The fundamental reason structural change strengthens self-enforcing
agglomeration forces in the model (by increasing the share of mobile
expenditure in the economy) is robust to the specific assumptions
of Krugman's core-periphery model \citep{Krugman1991}. By contrast,
the role of trade and transportation costs on the rise and fall of
regional disparities is highly sensitive to some modeling choices
\citep{Davis1998,KrugmanVenables1995,Puga1999}. In Krugman's original
model, the source of immobile demand is agriculture. \citet{helpman1998size}
points out the fragility of the result by reversing the model assumptions.
He introduces housing as the source of immobile demand. and shows
that in this alternative setting dispersion forces dominate agglomeration
forces when inter-regional trade and transportation costs in manufacturing
are low enough, reversing Krugman's original result.

By contrast, the chain of logic that we emphasize---from productivity
growth to structural change and from structural change to regional
disparities---works in both Krugman's and Helpman's cases.\footnote{As we have discussed, we make this point formally in Appendix \ref{Appx:Helpman}.
Specifically, we show that the symmetric equilibrium is unstable for
arbitrarily large transportation costs if the manufacturing expenditure
share $m$ is large enough.} Specifically, it works under the assumption that the composite good
combining manufacturing and services is a luxury and the other good,
be it agriculture or housing, a necessity. Empirically, the income
elasticity of demand for both agriculture \citep*{Herrendorf2014}
and housing \citep*{combes2019} is lower than one.

Our model is also amenable to extensions with more than two sectors
since \textsc{hcd }preferences can accommodate any number of sectors
and goods. In particular, we view studying the full transition from
agriculture, to manufacturing, and to services as a promising avenue
for studying structural changes and economic geography jointly. For
a quantitative analysis, the model can be extended to multiple locations
and non-unitary price effects can be introduced, e.g., through a \textsc{nh-ces}
demand system. We anticipate that in both cases, the two-way interaction
between structural change and agglomeration that we uncover may become
more relevant as it triggers further structural change to services
or to additional locations.

\bibliographystyle{plainnat}
\bibliography{biblio_bmrn}

\clearpage{}

\part*{Appendix}

\appendix

\section{\protect\label{proofs:=000020existence=000020and=000020uniqueness}Proofs
for Section \ref{sec:Model}}

\subsection{Lemma \ref{lem:Regularity=000020conditions=000020--=000020consumption}.}
\begin{lem*}
\label{lem:Regularity=000020conditions=000020--=000020consumption-1}\textbf{Regularity
conditions}. Assume preferences obey equation \eqref{eq:HCD}, and
assume assumptions \ref{assu:omega} and \ref{assu:subsistence} hold;
then in equilibrium (i) $C_{i}>\gamma$ and (ii) $\nicefrac{\partial u}{\partial C_{i}}>0$
for all $i\in I$.
\end{lem*}
\begin{proof}
Part (i) is equivalent to showing that $C_{i}=\gamma$ implies the
violation of the conditions imposed in \ref{assu:omega}, namely,
that $P_{i}\gamma\geq\uline{s}y$. Minimizing total expenditure $\sum_{i}P_{i}C_{i}$
to reach an arbitrary level of utility $u>0$ yields the following
first-order condition for $C_{i}$:
\begin{equation}
P_{i}-\frac{\theta\omega_{i}}{Ci}\geq0,\quad C_{i}\geq\gamma,\quad\left(P_{i}-\frac{\theta\omega_{i}}{C_{i}}\right)\left(C_{i}-\gamma\right)=0,\label{eq:foc}
\end{equation}
where $\theta$ is the Lagrange multiplier associated with the constraint
in equation \eqref{eq:HCD}. Thus, $C_{i}=\gamma$ requires $P_{i}\gamma\geq\theta\omega_{i}$,
and we are set to showing $P_{i}\gamma\geq\theta\omega_{i}\Rightarrow P_{i}\gamma\geq\uline{s}y$
to establish part (i). Note first that $P_{i}\gamma<\uline{s}y$ for
all $i\in{\cal N}$ implies that $y$ is above the level of expenditure
required to cover the subsistence level
\begin{align*}
\sum_{i}P_{i}\gamma & <\sum_{i}\uline{s}y<y,
\end{align*}
where the second inequality follows from the fact that the sum of
all goods minimum expenditure shares must be less than to the sum
of their actual expenditure shares which sum to one. Thus, there exists
a subset of goods $I_{+}\subseteq I$ such that $C_{j}>\gamma_{}$
and $P_{j}C_{j}=\theta\omega_{j}$ for all $j\in I_{+}$. Denote the
complementary subset of goods where $C_{i}=\gamma$ in equilibrium
by $I_{-}$. We can then write:
\begin{align*}
y & =\sum_{J\in I}P_{J}C_{J}\\
 & =\sum_{i\in I_{-}}P_{i}\gamma+\sum_{j\in I_{+}}P_{j}C_{j}\\
 & =\sum_{i\in I_{-}}P_{i}\gamma+\theta\sum_{j\in I_{+}}\omega_{j},
\end{align*}
where the first equality equalizes income with expenditure, the second
exploits the partition into goods consumed at the subsistence level
and those consumed above it, and the third exploits the first order
condition for the intensive margin of the latter category. Solving
for $\theta$, we obtain:
\[
\theta=\frac{y-\sum_{i\in I_{-}}P_{i}\gamma}{\sum_{j\in I_{+}}\omega_{j}}.
\]
Plugging this expression into the first order condition for goods
$i\in I_{-}$ yields:
\[
P_{i}\gamma\geq\frac{y-\sum_{i\in I_{-}}P_{i}\gamma}{\sum_{j\in I_{+}}\omega_{j}}\omega_{i}.
\]
Multiplying both sides by $\sum_{j\in I_{+}}\omega_{j}$ and summing
over goods $i\in I_{-}$ yields:
\[
\left(\sum_{j\in I_{+}}\omega_{j}\right)\sum_{i\in I_{-}}P_{i}\gamma\geq\left(\sum_{i\in I_{-}}\omega_{i}\right)\left(y-\sum_{i\in I_{-}}P_{i}\gamma\right),
\]
which we may rewrite as 
\[
\left(\sum_{j\in I_{+}}\omega_{j}+\sum_{i\in I_{-}}\omega_{i}\right)\sum_{i\in I_{-}}Pi\gamma\geq\left(\sum_{i\in I_{-}}\omega_{i}\right)y.
\]
Rearranging and simplifying, we obtain:
\[
\sum_{i\in I_{-}}Pi\gamma\geq\sum_{i\in I_{-}}\frac{\omega_{i}}{\omega}y,
\]
which implies
\begin{align*}
\exists i\in I_{-}:\quad P_{i}\gamma & \geq\frac{\omega_{i}}{\omega}y>\uline{s}y,
\end{align*}
where the second inequality follows from part (i) of \ref{assu:omega}.
Thus, if the condition in part (iii) of \eqref{assu:omega} holds,
then $C_{i}>\gamma$ for all $i\in I$. This step completes the proof
for part (i).

Part (ii). Totally differentiating equation \eqref{eq:HCD} and rearranging
it yields
\[
\left[\frac{1}{u}-\sum_{i}\omega_{i}^{\prime}\left(u\right)\ln\left(\frac{\omega(u)}{\omega_{i}(u)}\frac{C_{i}}{\gamma}\right)\right]\mathrm{d}u=\sum_{i}\omega_{i}\left(u\right)\mathrm{d}\ln C_{i}.
\]
The term in brackets on the left-hand-side of the expression above
is positive by $\omega_{i}^{\prime}<0$, $C_{i}>\gamma$, and $\omega_{i}(u)<\omega(u)$.
Therefore, $u$ is increasing in $C_{i}$ by $\omega_{i}>0$.
\end{proof}

\subsection{Lemma \ref{lem:Indirect=000020Utility}}
\begin{lem*}
\textbf{\label{lem:Indirect=000020Utility-1}Indirect utility function}.
Assume assumptions \ref{assu:subsistence} and \ref{assu:omega} hold.
Then the indirect utility $V\left(\{P_{i}\},y\right)$ associated
with the preferences in equation \eqref{eq:HCD} can be written implicitly
as the fixed point for $v$ of the following expression:
\[
\ln v=\sum_{i}\omega_{i}\left(v\right)\ln\left(\frac{y}{\gamma P_{i}}\right)..
\]
This fixed point exists and is unique. Furthermore, $V\left(\{P_{i}\},y\right)$
is increasing in $y$, decreasing in $P_{i}$ for all $i\in I$, and
homogeneous of degree zero in its arguments (and hence $V$ is a proper
indirect utility function).
\end{lem*}
\begin{proof}
Totally differentiating equation \eqref{eq:Indirect=000020Utility=000020--=000020intermediate=000020step}
and rearranging yield:
\[
\left[\frac{1}{v}-\sum_{i}\omega_{i}^{\prime}\left(v\right)\ln\left(\frac{y}{\gamma P_{i}}\right)\right]\mathrm{d}v=\omega\left(v\right)\mathrm{d}\ln y-\sum_{i}\omega_{i}\left(v\right)\mathrm{d}\ln P_{i}.
\]
The term in brackets on the left-hand-side of the expression above
is positive by assumption \ref{assu:omega}. Thus, $v$ is increasing
in $y$ and decreasing in $P_{i}$ by $\omega_{i}>0$.
\end{proof}

\subsection{Lemma \ref{lem:Expenditure=000020Function}}
\begin{lem*}
\textbf{\label{lem:Expenditure=000020Function-1}Expenditure function}.
Assume assumptions \ref{assu:subsistence} and \ref{assu:omega} hold.
Then the expenditure function $E\left(u,P_{A},P_{M}\right)$ associated
with the preferences in equation \eqref{eq:HCD} can be written as:
\begin{align}
\ln E\left(\{P_{i}\},u\right) & =\ln\gamma+\frac{\ln u}{\omega\left(u\right)}+\sum_{i}\frac{\omega_{i}\left(u\right)}{\omega(u)}\mathrm{d}\ln P_{i}\label{eq:Expenditure=000020Function=000020E-1}
\end{align}
The function $E$ is increasing in its arguments and it is homogeneous
of degree one in prices (and hence is a proper expenditure function).
\end{lem*}
\begin{proof}
The proof is analogous to the proofs to lemmas \ref{lem:Regularity=000020conditions=000020--=000020consumption}
and \ref{lem:Indirect=000020Utility-1}, and hence is omitted.
\end{proof}

\subsection{Proposition \ref{prop:Structural=000020change}}
\begin{prop*}
\label{prop:Structural=000020change-1}\textbf{Engel curves and structural
change}. Assume assumption \ref{assu:Engel} holds. Then (i) manufactures
are a luxury good and the agricultural good is a necessity, and (ii)
the allocation of workers across sectors follows the expenditure shares.
Hence $L_{M}$ is increasing in the level of utility $u$ in equilibrium.
\end{prop*}
\begin{proof}
(i) Manufactures are a luxury good if and only if the income elasticity
of demand is larger than one. With only two goods in the economy,
this requirement is equivalent to saying that $m=\frac{\omega_{M}}{\omega_{A}+\omega_{M}}$
is increasing in $u$, which can happen if and only if $\frac{\omega_{M}}{\omega_{A}}$
is increasing in $u$. (ii) By inspection, $L_{M}$ is increasing
in $m$, which is increasing in $u$ by part (i).
\end{proof}

\subsection{Lemma \ref{lem:Sufficient=000020condition}}
\begin{lem*}
\label{lem:Sufficient=000020condition-1}\textbf{Sufficient regularity
conditions}. If parameter values are such that the following inequality
holds:
\begin{equation}
\frac{\gamma\sigma\tau}{\uline{s}}<1,\label{eq:sufficient=000020condition=000020--=000020fundamentals-1}
\end{equation}
then the sufficient condition in assumption \ref{assu:omega} holds.
\end{lem*}
\begin{proof}
Equation \eqref{eq:P_M^s} implies $P_{M}\leq\tau$, and $P_{A}=1$
by equation \eqref{eq:w^r=00003Dalpha}. Turn to equation \eqref{eq:sales},
and rewrite it as follows:
\begin{align*}
\sigma\pi^{r} & =\frac{\lambda^{r}m\left(u^{r}\right)\pi^{r}+\frac{\alpha}{2}}{\lambda^{r}+\left(1-\lambda^{r}\right)\tau^{1-\sigma}}+\frac{\left(1-\lambda^{r}\right)m\left(u^{s}\right)\pi^{s}+\frac{\alpha}{2}}{\left(1-\lambda^{r}\right)\tau^{\sigma-1}+\lambda^{r}}\\
 & \geq\frac{\alpha}{2}\left[\frac{1}{\lambda^{r}+\left(1-\lambda^{r}\right)\tau^{1-\sigma}}+\frac{1}{\left(1-\lambda^{r}\right)\tau^{\sigma-1}+\lambda^{r}}\right].
\end{align*}
From $\sigma,\tau>1$ and $0\leq\lambda^{r}\leq1$, it follows that
the term in the square bracket is larger than 2. Then, the earnings
of an entrepreneur are larger than $\pi>\frac{\alpha}{\sigma}$, and
in turn the right-hand side of expression \eqref{eq:sufficient=000020condition=000020--=000020fundamentals-1}
is a lower bound for the right-hand side of the inequality in assumption
\ref{assu:omega}. Hence, if the former inequality holds, so does
the latter.
\end{proof}

\section{Proofs for Section \ref{sec:Equilibrium}}

\subsection{Lemma \ref{lem:Sustain=000020point-2}}
\begin{lem*}
\label{lem:Sustain=000020point-1}\textbf{Sustain point}. Let $\alpha_{1}\left(\sigma\right)>0$
be defined implicitly by $m\left(\alpha_{1},\sigma,1\right)=\sigma-1$
when $\sigma\in\left(1,2\right]$, with $\alpha_{1}^{\prime}>0$ by
equation \eqref{eq:CP=000020alpha=000020sigma}. (i) If $\sigma\in\left(1,2\right]$
and $0<\alpha\leq\alpha_{1}\left(\sigma\right)$, then $\Omega^{s}\leq\Omega^{C}$
if $\lambda^{r}=1$ for any value of $\tau$ in $\left[1,\tau_{1}\right]$.
(ii) If $\sigma\in\left(1,2\right]$ and $\alpha>\alpha_{1}\left(\sigma\right)$,
then $\Omega^{s}\leq\Omega^{C}$ if $\lambda^{r}=1$ for any $\tau\geq1$.
(iii) If $\sigma>2$, then $\Omega^{s}\leq\Omega^{CP}$ if $\lambda^{r}=1$
for any value of $\tau$ in $\left[1,\tau_{1}\right]$..
\end{lem*}
\begin{proof}
Rewrite equation \eqref{eq:sustain=000020plane} as
\[
f\left(\alpha,\sigma,\tau\right)\equiv1-B_{1}\left(\alpha,\sigma\right)\tau^{b\left(\alpha,\sigma\right)}+B_{2}\left(\alpha,\sigma\right)\tau^{2\left(1-\sigma\right)}=0,\qquad B_{1},B_{2}>1.
\]
By Descartes's generalization of Laguerre's rule of signs, this equation
admits one positive root for $\tau$ if $b\left(\alpha,\sigma\right)>0$,
as is in cases (i) and (iii) of the lemma, and it admits two positive
roots if $b\left(\alpha,\sigma\right)<0$, as in case (ii). In the
former case, it is easily verified by inspection of equation \eqref{eq:sustain=000020plane}
that this root is $\tau=1$, and, from equations \eqref{eq:v=000020Core},
\eqref{eq:v=000020periphery-1}, and \eqref{eq:sustain=000020plane},
that $\Omega^{C}>\Omega^{s}$ whenever $\tau>1$. If instead $b\left(\alpha,\sigma\right)<0$,
then we can establish with some algebra that $\lim_{\tau\rightarrow\infty}f\left(\alpha,\sigma,\tau\right)=0$,
$\lim_{\tau\rightarrow1}f\left(\alpha,\sigma,\tau\right)=-\infty$,
$\lim_{\tau\rightarrow\infty}f^{\prime}\left(\alpha,\sigma,\tau\right)=+\infty$,
and $\lim_{\tau\rightarrow1}f^{\prime}\left(\alpha,\sigma,\tau\right)<0$,
hence there exists a unique $\tau_{1}\in\left(1,\infty\right)$ such
that $f\left(\alpha,\sigma,\tau_{1}\right)=0$.
\end{proof}

\subsection{Lemma \ref{lem:Break=000020point-2}}
\begin{lem*}
\label{lem:Break=000020point-1}\textbf{Break point}. Let $\alpha_{\infty}\left(\sigma\right)>0$
be defined implicitly by $m\left(\alpha_{\infty},\sigma,\infty\right)=\sigma-1$
when $\sigma\in\left(1,2\right]$, with $\alpha_{\infty}^{\prime}>0$
by equation \eqref{eq:CP=000020alpha=000020sigma}. (i) If $\sigma\in\left(1,2\right]$
and $0<\alpha<\alpha_{\infty}\left(\sigma\right)$, then, the symmetric
equilibrium is unstable for any value of $\tau$ in $\left[1,\tau_{0}\right]$.
(ii) If $\sigma\in\left(1,2\right]$ and $\alpha>\alpha_{\infty}\left(\sigma\right)$,
then the symmetric equilibrium is unstable for any $\tau\geq1$. (iii)
If $\sigma>2$, then the symmetric equilibrium is unstable for any
value of $\tau$ in $\left[1,\tau_{0}\right]$.
\end{lem*}
\begin{proof}
The formal analysis of the stability of the break point was initiated
by \citet{Puga1999} and is now standard; see e.g., \citet*{BaldwinForslidMartinOttavianoRobertnicoud},
\citet*{FujitaKV1999}, or \citet{RobertNicoud2005}. Specifically,
we totally differentiate equations \eqref{eq:P_M^s}, \eqref{eq:sales},
and \eqref{eq:Expenditure=000020=00003D=000020income} around $\lambda^{r}=\lambda^{s}=\frac{1}{2}$,
and to check for conditions under which $\frac{\partial\Omega^{r}}{\partial\lambda^{r}}<0$
at this symmetric equilibrium; the symmetric equilibrium is said to
be stable under these conditions. If there exists a combination of
parameters $\left(\alpha,\sigma,\tau\right)$ in the admissible range
$\left(1,\infty\right)^{3}$ such that $\frac{\partial\Omega^{r}}{\partial\lambda^{r}}=0$,
then this combination satisfies equation \eqref{eq:break=000020point}.
The only difference between our model and the earlier literature is
that expenditure share on manufactures, $m$, is an endogenous variable
here while it is a parameter in the \textsc{neg} there. However, the
break point is defined as the parameter configuration such that the
change in utility following the move of an atomistic entrepreneur
is exactly zero. Since the change in utility is zero, it follows from
$m^{\prime}\left(\Omega\right)>0$ that $m$ is unchanged, too. The
rest of the proof is similar to that of lemma \ref{lem:Sustain=000020point-2},
and hence is omitted here.
\end{proof}

\subsection{Lemma \ref{lem:break=000020sustain=000020points=000020ranking}}
\begin{lem*}
\label{lem:break=000020sustain=000020points=000020ranking-1}\textbf{Equilibriums
ranking}. (i) Given parameter values $\left\{ \alpha,\phi,\sigma\right\} $,
the utility of entrepreneurs is strictly higher at the core-periphery
equilibrium than at the symmetric equilibrium:
\begin{equation}
\forall\left\{ \alpha,\sigma,\tau\right\} :\quad\Omega^{B}\left(\alpha,\sigma,\tau\right)<\Omega^{C}\left(\alpha,\sigma,1\right).\label{eq:utility=000020ranking-1}
\end{equation}
This result implies in turn $m^{B}<m^{C}$, and $\pi^{B}<\pi^{C}$.
(ii) Given parameter values $\left\{ \alpha,\sigma\right\} $, $\tau_{1}>\tau_{0}$.
(iii) Any spatial equilibrium other than $\lambda\in\left\{ 0,\frac{1}{2},1\right\} $,
if it exists at all, is not stable.
\end{lem*}
\begin{proof}
(i) Note that $\Omega^{B}=\Omega^{C}$ if $\tau=1$. Thus, to show
$\Omega^{B}<\Omega^{C}$ it suffices to show that $\Omega^{B}$ is
continuously decreasing in $\tau$. Totally differentiating equation
\eqref{eq:v=000020symmetric-1} yields
\begin{align*}
\mathrm{d}\ln\Omega^{B}-m^{B}\left[\frac{1}{\sigma-m^{B}}-\frac{1}{\sigma-1}\ln\left(\frac{2}{1+\tau^{1-\sigma}}\right)\right]\mathrm{d}\ln m^{B} & =\mathrm{d}\ln\alpha-m^{B}\frac{\tau^{1-\sigma}}{1+\tau^{1-\sigma}}\mathrm{d}\ln\tau,
\end{align*}
which we may rewrite as 
\begin{equation}
\mathrm{d}\ln\Omega^{B}\left\{ 1-m^{B}\left[\frac{1}{\sigma-m^{B}}-\frac{1}{\sigma-1}\ln\left(\frac{2}{1+\tau^{1-\sigma}}\right)\right]\frac{\partial\ln m^{B}}{\partial\ln\Omega^{B}}\right\} =\mathrm{d}\ln\alpha-m^{B}\frac{\tau^{1-\sigma}}{1+\tau^{1-\sigma}}\mathrm{d}\ln\tau.\label{eq:Omega^B=000020total=000020diff}
\end{equation}
Consider the following term in the braces of the expression above:
\begin{align*}
m^{B}\left[\frac{1}{\sigma-m^{B}}-\frac{1}{\sigma-1}\ln\left(\frac{2}{1+\tau^{1-\sigma}}\right)\right] & \leq\frac{1}{\sigma-1}-\frac{1}{\sigma-1}\ln\left(\frac{2}{1+\tau^{1-\sigma}}\right)\\
 & \leq\frac{1}{\sigma-1}.
\end{align*}
Under assumption \eqref{assu:regularity-1}, then, the term inside
braces in equation \eqref{eq:Omega^B=000020total=000020diff} is positive,
and $\Omega^{B}$ is decreasing in $\tau$, and hence $\Omega^{B}<\Omega^{C}$,
as was to be shown. (ii) Consider the break point in equation \eqref{eq:v=000020symmetric-1},
and evaluate this expression for $m^{B}\left(\alpha,\sigma,\tau\right)=m^{C}$.
Let 
\[
\left(\tilde{\tau}_{0}\right)^{1-\sigma}\equiv\frac{\sigma-1-m^{C}}{\sigma-1+m^{C}}\frac{\sigma-m^{C}}{\sigma+m^{C}},
\]
which is the break point treating $m$ as a parameter, which we fix
at $m^{C}.$ In this case, \citet{RobertNicoud2005} shows that $\tilde{\tau}_{0}<\tau_{1}$
holds. Thus, we are left to showing that $\tau_{0}<\tilde{\tau}_{0}$
holds. To see that this is indeed the case, observe that it follows
from $\Omega^{B}<\Omega^{C}$ and assumption \ref{assu:Engel} that
$m^{B}$ is decreasing in $\tau$, and that $m^{B}<m^{C}$ holds.
Since $\tilde{\tau}_{0}$ is increasing in $m$ by inspection of the
expression above (recall $\sigma>1$), it follows that $\tau_{0}<\tilde{\tau}_{0}$
holds, and hence $\tau_{0}<\tau_{1}$ holds, as was to be shown. (iii)
\citet{RobertNicoud2005} shows that this result is a corollary to
part (ii) of this lemma.
\end{proof}

\subsection{Proposition \ref{prop:Break=000020and=000020Sustain=000020points}}
\begin{prop*}
\label{prop:Break=000020and=000020Sustain=000020points-1}If $\sigma\in\left(1,2\right]$,
then $\exists\alpha_{1}\left(\sigma\right),\alpha_{\infty}\left(\sigma\right)>0$
defined implicitly as $m^{C}\left(\alpha_{1},\sigma,1\right)=\sigma-1$
and $m^{B}\left(\alpha_{\infty},\sigma,\infty\right)$, respectively,
with $0<\alpha_{1}\left(\sigma\right)<\alpha_{\infty}\left(\sigma\right)$.
(i) The symmetric equilibrium is the unique stable equilibrium configuration
in either of the following cases: if $\sigma>2$ and $\tau>\tau_{1}\left(\alpha,\sigma\right)$;
or if $\sigma\in\left(1,2\right]$, $0<\alpha<\alpha_{1}\left(\sigma\right)$,
and $\tau>\tau_{1}\left(\alpha,\sigma\right)$. (ii) The core-periphery
configuration is the unique stable equilibrium in any of the following
cases: if $\sigma>2$ and $\tau\leq\tau_{0}\left(\alpha,\sigma\right)$;
if $\sigma\in\left(1,2\right]$, $0<\alpha<\alpha_{\infty}\left(\sigma\right)$,
and $\tau\leq\tau_{0}\left(\alpha,\sigma\right)$; if $\sigma\in\left(1,2\right]$
and $\alpha\geq\alpha_{\infty}$$\left(\sigma\right)$. (iii) The
symmetric equilibrium and the core-periphery pattern are both stable
equilibrium configurations in any of the following cases: if $\sigma>2$
and $\tau_{0}\left(\alpha,\sigma\right)<\tau\leq\tau_{1}\left(\alpha,\sigma\right)$;
or if $\sigma\in\left(1,2\right]$, $0<\alpha<\alpha_{1}\left(\sigma\right)$,
and $\tau_{0}\left(\alpha,\sigma\right)<\tau\leq\tau_{1}\left(\alpha,\sigma\right)$;
or if $\sigma\in\left(1,2\right]$, $\alpha_{1}\left(\sigma\right)<\alpha<\alpha_{\infty}\left(\sigma\right)$,
and $\tau>\tau_{0}\left(\alpha,\sigma\right)$.
\end{prop*}
\begin{proof}
The proof of this proposition is an immediate consequence of lemmas
\ref{lem:Sustain=000020point-2}, \ref{lem:Break=000020point-2},
and \ref{lem:break=000020sustain=000020points=000020ranking}.
\end{proof}

\section{Proofs for Section \ref{sec:Structural=000020Change=000020Spatial=000020Disparities}}

\subsection{Proposition \ref{prop:Growth=000020and=000020Structural=000020change-1}}
\begin{prop*}
\label{prop:Growth=000020and=000020Structural=000020change-1}\textbf{Growth
and structural change}. Assume labor productivity $\alpha$ increases
over time. Then, given the spatial equilibrium $j\in\left\{ B,CP\right\} $,
national income $Y$, per-capita income $\frac{Y}{2}$, the share
of manufacturing $\tilde{m}$, and the earnings premium of entrepreneurs
$\tilde{\pi}$ all increase over time.
\end{prop*}
\begin{proof}
By inspection of equations \eqref{eq:National=000020income}, \eqref{eq:m=000020share=000020income},
and \eqref{eq:premium}.
\end{proof}

\subsection{Proposition \ref{prop:Structural=000020change=000020and=000020spatial=000020disparities-1}}
\begin{prop*}
\label{prop:Structural=000020change=000020and=000020spatial=000020disparities-1}\textbf{Structural
change and spatial disparities}. Assume labor productivity $\alpha$
increases over time, all other things equal. Then, starting from an
initially low value $\alpha<\alpha_{1}$, (i) structural change takes
place: expenditure and employment are reallocated from agriculture
to manufacturing; (ii) regional disparities may (if $\alpha_{1}<\alpha<\alpha_{0}$)
and will (if $\alpha_{0}<\alpha$) emerge in equilibrium; (iii) regional
disparities reinforce structural change; and (iv) structural change
leads to an increase in the earnings premium of entrepreneurs.
\end{prop*}
\begin{proof}
In the text above the proposition.
\end{proof}

\section{\protect\label{Appx:Helpman}Encompassing Competition for Land}

In the paper we assert that demand-driven structural change contributes
to the emergence of regional disparities as long as the income elasticity
of the demand for the combination of tradable services and manufacturing
goods is larger than one (and the income elasticity of demand for
agricultural goods is lower than one). Here we build a portmanteau
model nesting both Krugman and Helpman's mechanisms, of which the
baseline model in our main text is a special case.

\subsection{Model}

We add structures to the in section \ref{sec:Model}. Specifically,
we now have 3 sectors:
\begin{itemize}
\item Agriculture. This sector is like in section \ref{sec:Model}. It is
competitive, it produces using labor only, and its output is freely
traded.
\item Manufacturing. As in section \ref{sec:Model}, entrepreneurs are the
fixed factor. Now, the variable factor is a Cobb-Douglas combination
of labor and structures, where $\eta\in\left[0,1\right]$ is Cobb
Douglas share of structures. The model nests our Krugman-style model
as a special case, with $\eta=0$, and a variant of Helpman's model
in the polar case $\eta=1$. Like in in section \ref{sec:Model},
this sector is monopolistically competitive and its output is traded
at an iceberg cost $\tau$.
\item Structures (factories). Structures are in fixed supply and are not
traded. They are used in manufacturing production only (it is conceptually
straightforward to also allow for housing consumption).
\end{itemize}
The main change with our baseline model is that we need impose the
equilibrium condition for the market of structures:
\begin{equation}
\Lambda\rho^{r}=\eta\frac{\sigma-1}{\sigma}R^{r}=\eta\left(\sigma-1\right)\lambda^{r}\pi^{r}.\label{eq:Structures=000020Equilibrium}
\end{equation}
Above, $\rho^{r}$ is the equilibrium unit price of structures and
$\Lambda$ is the fixed supply of land. We set $\Lambda=\eta\left(\sigma-1\right)$
by choice of units, thus $\rho^{r}=\lambda^{r}\pi^{r}$.

Manufacturing sales in equation \eqref{eq:sales} are now equal to
\[
\sigma\pi^{r}=\left(\lambda^{r}\pi^{r}\right)^{\eta\left(1-\sigma\right)}\left\{ \frac{E^{r}}{\Delta^{r}}+\tau^{1-\sigma}\frac{E^{s}}{\Delta^{s}}\right\} ,
\]
where 
\[
E^{r}=\left[m\left(\Omega^{r}\right)+\eta\left(\sigma-1\right)\right]\lambda^{r}\pi^{r}+\frac{\alpha}{2}
\]
denotes manufacturing expenditures, and 
\[
\Delta^{r}\equiv\left(P^{r}\right)^{1-\sigma}=\lambda^{r}\left(\lambda^{r}\pi^{r}\right)^{\eta\left(1-\sigma\right)}+\tau^{1-\sigma}\lambda^{s}\left(\lambda^{s}\pi^{s}\right)^{\eta\left(1-\sigma\right)}.
\]

\subsection{Hat Algebra}

It proves useful to define $\hat{x}\equiv x^{r}/x^{s}$ for all $x=\lambda,\pi,\Delta,E,P$.
Taking ratios of the expressions above yields:
\begin{align}
\hat{\pi} & =\left(\hat{\lambda}\hat{\pi}\right)^{\eta\left(1-\sigma\right)}\frac{\hat{E}+\tau^{1-\sigma}\hat{\Delta}}{\tau^{1-\sigma}\hat{E}+\hat{\Delta}}\nonumber \\
\hat{E} & =\frac{\alpha+\left[m^{r}+\eta\left(\sigma-1\right)\right]2\lambda^{r}\pi^{r}}{\alpha+\left[m^{s}+\eta\left(\sigma-1\right)\right]2\lambda^{s}\pi^{s}}\nonumber \\
\hat{\Delta} & =\frac{\hat{\lambda}^{1+\eta\left(1-\sigma\right)}\hat{\pi}^{\eta\left(1-\sigma\right)}+\tau^{1-\sigma}}{\tau^{1-\sigma}\hat{\lambda}^{1+\eta\left(1-\sigma\right)}\hat{\pi}^{\eta\left(1-\sigma\right)}+1}.\label{eq:hat=000020system}
\end{align}
By the same token, if $m^{s}=m^{r}=m$, the ratio of real entrepreneur
earnings is equal to
\[
\hat{\Omega}=\hat{\pi}\hat{\Delta}^{\frac{m}{\sigma-1}}.
\]

\subsection{Stability of the Symmetric Equilibrium}

At the symmetric equilibrium, $\hat{x}=1$ for all $x$, and $m^{r}=m^{s}=m$,
some $m\in\left(0,1\right)$. In particular, the equilibrium earnings
are equal to 
\[
\pi=\frac{\alpha}{\sigma-\left[m+\eta\left(\sigma-1\right)\right]},
\]
which is positive by $\eta,m\in\left(0,1\right)$ and $\sigma>1$,
and increasing in $m$. We are interested in the stability of this
equilibrium. Let 
\[
Z\equiv\frac{1-\tau^{1-\sigma}}{1+\tau^{1-\sigma}}
\]
be an index of trade frictions. By inspection, $Z$ is monotonically
increasing in $\tau$ from $Z=0$ when $\tau=1$ (free trade) to $Z=1$
when $\tau\rightarrow\infty$ (autarky). Let also define the elasticity
function as follows:
\[
{\cal E}\left(x\right)\equiv\left.\frac{\mathrm{d}\ln\hat{x}}{\mathrm{d}\ln\hat{\lambda}}\right|_{\hat{\lambda}=1}.
\]
Using this compact notation, the symmetric equilibrium is said to
be stable if 
\begin{equation}
{\cal E}\left(\Omega\right)={\cal E}\left(\pi\right)+\frac{m}{\sigma-1}{\cal E}\left(\Delta\right)\label{eq:elasticity=000020Omega}
\end{equation}
is negative, and unstable otherwise.

Differentiating the system in equation \eqref{eq:hat=000020system},
we obtain
\begin{align*}
{\cal E}\left(\pi\right) & =-\eta\left(\sigma-1\right)\left[1+{\cal E}\left(\pi\right)\right]+Z\left[{\cal E}\left(E\right)-{\cal E}\left(\Delta\right)\right]\\
{\cal E}\left(E\right) & =2\tilde{m}\left[1+{\cal E}\left(\pi\right)\right]\\
{\cal E}\left(\Delta\right) & =Z\left\{ 1-\eta\left(\sigma-1\right)\left[1+{\cal E}\left(\pi\right)\right]\right\} ,
\end{align*}
where 
\[
\tilde{m}\equiv\frac{m+\eta\left(\sigma-1\right)}{\sigma}
\]
collects terms. By inspection, $\tilde{m}$ is increasing in $m$,
and it belongs to the unit interval.

Let us rewrite the system of equations in matrix notation:
\[
\left[\begin{array}{ccc}
1+\eta\left(\sigma-1\right) & -Z & Z\\
-2\tilde{m} & 1 & 0\\
-Z\eta\left(\sigma-1\right) & 0 & 1
\end{array}\right]\left[\begin{array}{c}
{\cal E}\left(\pi\right)\\
{\cal E}\left(E\right)\\
{\cal E}\left(\Delta\right)
\end{array}\right]=\left[\begin{array}{c}
-\eta\left(\sigma-1\right)\\
2\tilde{m}\\
Z\left[1-\eta\left(\sigma-1\right)\right]
\end{array}\right].
\]
Solving this system for ${\cal E}\left(\pi\right)$ and ${\cal E}\left(\Delta\right)$
and using equation \eqref{eq:elasticity=000020Omega} yields:
\begin{align}
{\cal E}\left(\Omega\right) & =\frac{-\eta\left(\sigma-1\right)+2\tilde{m}Z-\left[1-\eta\left(\sigma-1\right)\right]Z^{2}+\frac{m}{\sigma-1}Z\left(1-2\tilde{m}Z\right)}{1-2\tilde{m}Z+\eta\left(\sigma-1\right)\left(1-Z^{2}\right)}\nonumber \\
 & =\frac{-\eta\left(\sigma-1\right)\left(1-Z^{2}\right)+2\tilde{m}Z-Z^{2}+\frac{m}{\sigma-1}Z\left(1-2\tilde{m}Z\right)}{1-2\tilde{m}Z+\eta\left(\sigma-1\right)\left(1-Z^{2}\right)}\nonumber \\
 & =-1+\frac{1-Z^{2}+\frac{m}{\sigma-1}Z\left(1-2\tilde{m}Z\right)}{1-2\tilde{m}Z+\eta\left(\sigma-1\right)\left(1-Z^{2}\right)}.\label{eq:elasticity=000020omega=000020solution}
\end{align}
From the first line of equation \eqref{eq:elasticity=000020omega=000020solution},
we obtain the following results:
\begin{enumerate}
\item The symmetric equilibrium is stable for arbitrarily low trade costs
if $\eta>0$. Indeed,
\[
\left.{\cal E}\left(\Omega\right)\right|_{Z=0}=-\left(1+\frac{1}{\eta\left(\sigma-1\right)}\right)^{-1},
\]
which is unambiguously negative if $\eta>0$ (as in Helpman's model)
and zero if $\eta=0$ (as in Krugman's model).
\item The symmetric equilibrium is unstable for arbitrarily high trade costs
regardless of $\eta$ if expenditures on manufactures are high enough.
Indeed,
\[
\left.{\cal E}\left(\Omega\right)\right|_{Z=1}=\frac{m}{\sigma-1}-1,
\]
which is positive if and only if $m>\sigma-1$ (as in Krugman's model).
\end{enumerate}
\selectlanguage{american}%

\section{Data Construction\protect\label{sec:Data-Construction}}

We merge data from the Groningen 10 sector database, which we take
from \citet{CominLashkariMestieri2021}, with the data from the urbanization
rates from the United Nations \href{https://data.unhabitat.org/}{Unhabitat Project}.
We drop from the resulting dataset Singapore and Hong Kong, since
they are essentially city-states. The resulting dataset contains 37
countries, and spans from 1960 to 2010. The panel is not balanced
though, the average number of years per country is 40.2, with a standard
deviation of 9.1 years. The list of countries and the corresponding
number of years with data follows: Argentina, 46 years of data; Bolivia,
44; Brazil, 12; Chile, 46; Colombia, 46; Costa Rica, 46; India, 35;
Indonesia, 45; Japan, 44; Korea, 43; Mexico, 46; Malaysia, 31; Peru,
15; Philippines, 35; Thailand, 46; Venezuela, 44; West Germany, 24;
Denmark, 36; Spain, 36; France, 36; Italy, 36; Netherlands, 36; Sweden,
36; USA, 46; UK, 46; Botswana, 43; Ethiopia, 50; Ghana, 31; Kenya,
42; Mauritius, 41; Malawi, 45; Nigeria, 51; Senegal, 41; Tanzania,
51; South Africa, 51; Zambia, 46.\selectlanguage{english}%

\end{document}